\documentclass[11pt]{article}
\usepackage{amsmath,amsfonts,amssymb,amsthm}
\usepackage{mathrsfs}
\usepackage{latexsym}
\usepackage[english]{babel}

\usepackage{graphicx}
\usepackage[usenames,dvipsnames]{xcolor}

\renewenvironment{proof}[1][\proofname]{{\bfseries #1.} }{\qed}

\def\Cov{{\rm Cov\,}}

\newcommand{\field}[1]{\mathbb{#1}}

\newcommand{\R}{\field{R}}

\newcommand{\Var}{{\rm Var}}
\newcommand{\Corr}{{\rm Corr}}

\newcommand{\e}{{\rm e}}

\newcommand{\eps}{\varepsilon}

\def\authors#1{{ \begin{center} #1 \vspace{0pt} \end{center} } \smallskip}

\def\title#1{{\huge\bf  \begin{center} #1 \vspace{0pt} \end{center}  } \smallskip}

\def\E{{\mathbb{ E}}}

\def\P{{\mathbb{P}}}

\def\paref#1{(\ref{#1})}

\newtheorem{theorem}{Theorem}[section]
\newtheorem{proposition}[theorem]{Proposition}
\newtheorem{lemma}[theorem]{Lemma}
\newtheorem{corollary}[theorem]{Corollary}

\newtheorem{remark}[theorem]{Remark}

\begin{document}

\date{Nov 2018}

\title{{\bf On the correlation between nodal and boundary lengths for random spherical harmonics}}
\authors{\large\sc Domenico Marinucci and  Maurizia Rossi}
%

\begin{abstract}

We study the correlation between the nodal length of random spherical
harmonics  and the measure of the boundary for excursion sets at any non-zero
level. We show that the correlation is asymptotically zero, while the
partial correlation after controlling for the random $L^2$-norm on the sphere of the
eigenfunctions is asymptotically one.

\smallskip


\smallskip

\end{abstract}

\section{Introduction and main result}

\subsection{Random spherical harmonics}

On the unit two-dimensional sphere $\mathbb{S}^{2},$ let us consider the
Helmholtz equation
\[
\Delta _{\mathbb{S}^{2}}f_{\ell }+\lambda _{\ell }f_{\ell }=0\text{ },\hspace{%
1cm}f_{\ell }:\mathbb{S}^{2}\rightarrow \mathbb{\mathbb{R}},
\]%
where
\[
\Delta _{\mathbb{S}^{2}}=\frac{1}{\sin \theta }\frac{\partial }{\partial
\theta }\left (\sin \theta \frac{\partial }{\partial \theta } \right )+\frac{1}{\sin
^{2}\theta }\frac{\partial }{\partial \varphi ^{2}}
\]%
is the Laplace-Beltrami operator on $\mathbb{S}^{2}$ in spherical coordinates $(\theta,\varphi)$ and $\left\{ \lambda
_{\ell }=\ell (\ell +1)\right\}_{\ell\in \mathbb N}$ represent the set of
eigenvalues of $-\Delta_{\mathbb{S}^{2}}$. For any $\lambda _{\ell }$, the corresponding eigenspace is the
$(2\ell +1)$-dimensional space of spherical harmonics of degree $\ell $; as
usual, we can choose an $L^{2}$-orthonormal basis $\left\{ Y_{%
\mathbb{\ell }m}\right\} _{m=-\ell ,\dots , \ell }$ \cite[\S 3.4]{MaPeCUP}, and focus, for $\ell\in \mathbb N^*$, on random
eigenfunctions of the form
\begin{equation}
f_{\ell }(x)=\sqrt{\frac{4\pi }{2\ell +1}}\sum_{m=-\ell }^{\ell }a_{\ell
m}Y_{\ell m}(x)\text{ ,}\qquad x\in \mathbb S^2. \label{fell}
\end{equation}%
Here the coefficients $\left\{ a_{\mathbb{\ell }m}\right\}_{m=-\ell ,\dots , \ell }$
are random variables (defined on some probability space $(\Omega, \mathcal F, \mathbb P)$) such that $a_{\ell,0}$ is a real standard Gaussian, and for $m\ne 0$ the $a_{\ell,m}$'s are standard complex Gaussians (independent of $a_{\ell,0}$) and independent save for the relation $\overline{a_{\ell,m}} =(-1)^m a_{\ell, -m}$ ensuring $f_\ell$ to be real valued -- it is immediate to see that
the law of the process $f_{\ell }$ is invariant with
respect to the choice of a $L^{2}$-orthonormal basis of eigenfunctions.

For every $\ell$, the spherical
random field $f_{\ell }$ is centred, Gaussian
and isotropic; from the addition theorem for spherical harmonics \cite%
[(3.42)]{MaPeCUP}, the covariance function is given by
\[
r_\ell(x,y):=\mathbb{E}[f_{\ell }(x)f_{\ell }(y)]=P_{\ell }(\cos d(x,y)),\qquad x,y\in \mathbb S^2,
\]%
where $P_{\ell }$ is the $\ell$-th Legendre polynomial, defined by Rodrigues'
formula, as
\[
P_{\ell }(t):=\frac{1}{2^{\ell }\ell !}\frac{d^\ell}{dt^{\ell }}(t^{2}-1)^{\ell }%
\text{ , }\qquad  t\in [-1, 1],
\]%
and $d(x,y)$ stands for the spherical geodesic distance between the points $x$ and $y$.
As discussed elsewhere
(see i.e., \cite{CMW, CaMar2016, CW}), random spherical harmonics
arise naturally from the spectral analysis of isotropic spherical random
fields or in the investigation of quantum chaos (see for instance \cite{MaPeCUP}, \cite{wigmanreview} for reviews). We may assume $f_\ell$ and $f_{\ell'}$ to be independent  random fields whenever $\ell\ne \ell'$.

In this paper, we shall focus on the boundary length of the excursion sets of $%
 f_{\ell }$ in (\ref{fell}), which is an a.s. smooth curve whose connected components are homeomorphic to the circle; i.e., for $u\in \mathbb{R}$, we will investigate the sequence of random variables
\[
\mathcal{L}_{\ell }(u):=\text{length}\left ( f_{\ell }^{-1}(u)\right ).
\]
If $u=0$ the random quantity $\mathcal{L}_{\ell }(0)$ is known as nodal length.

\subsubsection{Nodal length}\label{Snodal}

The sequence of random variables $\{\mathcal L_\ell(0)\}_{\mathcal \ell\in \mathbb N^*}$
has been the object of an enormous amount of activity, see i.e. \cite{berard, Wig, MRW, npr, Tod18}.
In particular, according to the celebrated Yau's conjecture \cite%
{Yau}, which has been proved for real analytic manifolds by Donnelly and
Fefferman, see \cite{DonnFeffa}, and more recently for smooth manifolds by
\cite{limonov}, for any \emph{realization} $f_{\ell }$ one has, for positive
constants $c,C>0$,
\[
c\sqrt{\lambda _{\ell }}\leq \mathcal{L}_{\ell }(0)\leq C\sqrt{\lambda
_{\ell }}
\]%
for every $\ell \in \mathbb N^*$.
For the ``typical" eigenfunction, fluctuations are indeed of smaller order;
more precisely, tighter bounds can be given in a probabilistic sense. In
fact, under Gaussianity the expected value of $\mathcal{L}_{\ell }(0)$ is
easily computed to be (for instance by the Gaussian Kinematic Formula \cite{adlertaylor}, see also \cite{berard})%
\[
\mathbb{E}\left[ \mathcal{L}_{\ell }(0)\right] =2\pi \sqrt{\frac{\lambda _{\ell
}}{2}}.
\]%
The computation of the variance is much more challenging, and was solved by
\cite{Wig}, where it is shown that, as $\ell \rightarrow \infty ,$%
\begin{equation}\label{varnodal}
\Var\left( \mathcal{L}_{\ell }(0)\right) =\frac{\log \ell }{32}+O(1)\text{ .%
}
\end{equation}
More recently, \cite{MRW} actually provided a stronger characterization of
the nodal length fluctuation around their expected value. More precisely,
they were able to establish the asymptotic equivalence (in the $L^{2}(\mathbb P
)$-sense) of the nodal length and the so-called sample trispectrum of $f_{\ell}$ i.e.,
integral over the sphere of $H_{4}(f_{\ell })$ that is the
fourth-order Hermite polynomial evaluated at the field itself (we recall that $%
H_{4}(t)=t^{4}-6t^{2}+3$). Indeed, let us define first the sequence of random variables
\begin{equation}\label{M}
\mathcal{M}_{\ell }:=-\frac{1}{4}\sqrt{\frac{\lambda_\ell}{2}}\frac{1}{4!}%
\int_{\mathbb{S}^{2}}H_{4}(f_{\ell }(x))\,dx,\qquad \ell \in \mathbb N^*;
\end{equation}
thanks to the properties of Hermite and Legendre polynomials it is easy to check that \cite{MW2011,MW2014}
\[
\mathbb{E}\left[ \mathcal{M}_{\ell }\right] =0,\qquad \Var\left (  \mathcal{M%
}_{\ell }\right ) =\frac{\log \ell }{32}+O_{\ell\to +\infty}(1)\text{ . }
\]%
Moreover, denote by $\widetilde{\mathcal{L}}_{\ell }(0)$ the standardized
nodal length, i.e.,
\[
\widetilde{\mathcal{L}}_{\ell }(0):=\frac{\mathcal{L}_{\ell }(0)-\mathbb{E}%
[\mathcal{L}_{\ell }(0)]}{\sqrt{\Var(\mathcal{L}_{\ell }(0))}}\text{ ,}
\]%
and likewise for $\mathcal{M}_{\ell }$, writing $\widetilde{\mathcal{M}}%
_{\ell }:=\frac{\mathcal{M}_{\ell }}{\sqrt{\Var(\mathcal{M}_{\ell })}}.$ It
is shown in \cite{MRW} that, as $\ell \rightarrow \infty ,$%
\begin{equation}\label{mrw}
\mathbb{E}\left[ \left | \widetilde{\mathcal{L}}_{\ell }(0)-\widetilde{\mathcal{M}}%
_{\ell } \right |^2 \right]=O\left (\frac{1}{\log \ell} \right )\quad \Longrightarrow \quad \widetilde{%
\mathcal{L}}_{\ell }(0)
=\widetilde{\mathcal{M}}_{\ell }+o_{\mathbb P}(1),
\end{equation}
where by $o_{\mathbb P}(1)$ we mean a sequence of random variables converging to zero in probability.
Note that, bearing in mind \cite[Theorem 1.1]{MW2014}, \paref{mrw} yields immediately a CLT for the normalized nodal length.
%
%
%

\subsubsection{Boundary length}\label{Sboundary}

A natural further
question to investigate is then how much the nodal length behaviour
characterizes the full geometry of eigenfunctions, i.e., the behaviour of
excursion sets for arbitrary levels $u\neq 0$ \cite{MW2011}. To this aim, let us first
recall that (see \cite{Wig, ROSSI2015}), for $u\in \mathbb R$ and $\ell\in \mathbb N^*$,
\begin{equation}\label{mediau}
\mathbb{E}\left[ \mathcal{L}_{\ell }(u)\right] =2\pi\sqrt{\frac{\lambda _{\ell
}}{2}}\,\e^{-u^{2}/2},
\end{equation}
and, for $u\ne 0$, as $\ell\to +\infty$,
\begin{equation}\label{varu}
\Var\left (  \mathcal{L}_{\ell }(u)\right ) \sim \frac{\pi^2}{2} u^4 \e^{-u^2} \cdot \ell.
\end{equation}
Let us now consider the random variable (recall that $H_2(t) = t^2 -1$ is the second Hermite polynomial)
\begin{equation}\label{D}
\begin{split}
\mathcal D_\ell(u) := & \frac{1 }{4}\sqrt{ \frac{\lambda _{\ell }}{2}}\, u^{2}\e^{-u^2/2} \int_{%
\mathbb{S}^{2}}H_{2}(f_{\ell }(x))\,dx.
\end{split}
\end{equation}
From \paref{D} we easily find (cf. \paref{varu})
\begin{equation}\label{varDell}
\Var( \mathcal D_\ell(u)) = \pi^2 \frac{\lambda_\ell}{2\ell +1} u^4 \e^{-u^2};
\end{equation}
let us now consider, for $u\ne 0$, the standardized
boundary length, i.e.,
\[
\widetilde{\mathcal{L}}_{\ell }(u):=\frac{\mathcal{L}_{\ell }(u)-\mathbb{E}%
[\mathcal{L}_{\ell }(u)]}{\sqrt{\Var(\mathcal{L}_{\ell }(u))}}\text{ ,}
\]%
and analogously $\widetilde{\mathcal{D}}_{\ell }(u):=\frac{\mathcal{D}_{\ell }(u)}{\sqrt{\Var(\mathcal{D}_{\ell }(u))}}.$
 It was shown in \cite{ROSSI2015} that, for $u\ne 0$, as $\ell\to +\infty$,
\begin{equation}\label{lengthu}
\E[|\mathcal L_\ell(u) - \E[\mathcal L_\ell(u)] - \mathcal D_\ell(u)|^2] =o(\ell) \quad \Rightarrow\quad  \widetilde{%
\mathcal{L}}_{\ell }(u)
=\widetilde{\mathcal{D}}_{\ell }(u)+o_{\mathbb P}(1).
\end{equation}
%
In \S \ref{Soutline} we will interpret the results recalled in \S\ref{Snodal} and \S\ref{Sboundary} for nodal and boundary lengths respectively in terms of so-called chaotic components.

\begin{remark}\rm\label{parseval}

i) Let us stress now that one of the main differences w.r.t. the nodal case \paref{mrw} is that, as argued in \cite{CaMar2016}, thanks to \paref{lengthu} the boundary length is asymptotically (as $\ell\to +\infty$) \emph{perfectly correlated}
with other geometric functionals, such as the area \cite{DI, MW2014}, the Euler-Poincar\'{e}
characteristic \cite{CaMar2016}, and the number of critical points \cite{CaMar2019} above any nonzero threshold;
knowledge of any of these quantities yields asymptotically the full information on the behaviour of all the
others.

ii) From \paref{D} it is convenient to note that
\[
\int_{\mathbb{S}^{2}}H_{2}(f_{\ell }(x))\,dx=\int_{\mathbb{S}^{2}}(f_{\ell
}^{2}(x)-1)\,dx=\left\Vert f_{\ell }\right\Vert_{L^2(\mathbb S^2)}^{2}-4\pi =\left\Vert
f_{\ell }\right\Vert_{L^2(\mathbb S^2)}^{2}-\mathbb{E}\left[\left\Vert f_{\ell }\right\Vert_{L^2(\mathbb S^2)}^{2}\right];
\]
moreover a simple application of Parseval's identity yields
\[
\left\Vert f_{\ell }\right\Vert_{L^2(\mathbb S^2)}^{2}=\frac{4\pi}{2\ell+1}\sum_{m=-\ell }^{\ell }
|a_{\ell m}|^{2}=4\pi\,\widehat{C}_{\ell }\text{ ,}
\]
where $\widehat{C}_{\ell }:=(2\ell +1)^{-1}\sum_{m=-\ell }^{\ell }
|a_{\ell m}|^{2}$ is usually denoted the sample power spectrum. Roughly speaking, the boundary length is asymptotically proportional to the (random)
sample norm of the eigenfunctions: from this point of view, it is more
natural to expect that it should not play a role in the behaviour of nodal
lines, which are clearly invariant to scaling factors.
\end{remark}

Note that ii) in Remark \ref{parseval} and \paref{lengthu} immediately give a CLT for the normalized boundary length.

\subsection{Main result}

The main result in this paper is the characterization of the correlation
structure between the nodal and non-nodal case. To do so, it is important to
recall the standard distinction between the classical correlation
coefficient between two (finite-variance) random variables $X,Y$, which of
course is given by
\[
\rho (X,Y):=\frac{\Cov(X,Y)}{\sqrt{\Var(X)\Var(Y)}}\in[-1,1]\text{ ,}
\]%
and the Partial Correlation Coefficient (conditional on the finite variance
random variable $Z$), i.e.%
\begin{equation}\label{partialcorr}
\rho _{Z}(X,Y):=\rho({X^*},{Y^*}),
\end{equation}
where ${X^*},{Y^*}$ are the ``residuals" after projecting $%
X,Y$ on the ``explanatory variable" $Z$, i.e. for a finite-variance random variable $W$
\begin{equation}\label{explanatory}
{W^*}:=(W-\mathbb{E}[W])-\frac{\Cov(W,Z)}{\Var(Z)}(Z-\mathbb{E}[Z]%
\mathbb{)}.\ 
\end{equation}
As well-known, $\rho _{Z}(X,Y)$ admits a standard interpretation as a
measure of the (linear) dependence between $X$ and $Y,$ after we got rid of
the common components depending on $Z.$
Our main result is obtained by taking $X,Y$ $\ $to be the boundary lengths
at different levels and $Z=\left\Vert f_{\ell }\right\Vert
_{L^{2}(\mathbb S^{2})}^{2}=4\pi\,\widehat{C}_{\ell }$ to be the random $L^{2}$%
-norm of the eigenfunctions $f_{\ell }$:

\begin{theorem}\label{mainth}
As $\ell \rightarrow \infty ,$ for all $u_{1}\neq 0$
\begin{equation}\label{eqteo1}
\lim_{\ell \rightarrow \infty }\rho \left( \mathcal{L}_{\ell }(u_{1}),%
\mathcal{L}_{\ell }(u_{2})\right) =\left\{
\begin{array}{c}
1,\qquad \text{if } u_{2}\neq 0, \\
0,\qquad \text{if } u_{2}=0;%
\end{array}%
\right.
\end{equation}
on the other hand
\begin{equation}\label{eqteo2}
\lim_{\ell \rightarrow \infty }\rho _{\left\Vert f_{\ell }\right\Vert
_{L^{2}(\mathbb S^{2})}^{2}}\left( \mathcal{L}_{\ell }(u_{1}),\mathcal{L}_{\ell }(u_{2})\right) =1,\quad
\forall u_{1},u_{2}\in \mathbb{R}.
\end{equation}
\end{theorem}

As an easy corollary of this theorem, we obtain the following joint
weak-convergence results for the normalized boundary and nodal length.

\begin{corollary}\label{cor1}
As $\ell \rightarrow \infty ,$ for all $u\neq 0$%
\[
\left( \widetilde{\mathcal{L}}_{\ell }(u),\widetilde{\mathcal{L}}_{\ell
}(0)\right) \rightarrow _{d}(Z_{1},Z_{2})\text{ ,}
\]%
where $(Z_{1},Z_{2})$ denotes a bivariate vector of standard, independent
Gaussian variables.
\end{corollary}
Corollary \ref{cor1} states that the limiting distribution of the nodal
and boundary length (at non-zero level) are independent, in the
high energy limit. As motivated above, this results is substantially
spurious, as it depends crucially on the dominant role played in the
boundary length behaviour by the random energy of the eigenfunction. Taking
the effect into account, the landscape changes entirely; indeed, consider
the ``regression residual"%
\begin{equation}\label{resC}
\mathcal{L}_{\ell |\widehat{C}_{\ell }}(u):=\mathcal{L}_{\ell }(u)-\mathbb{E}
[\mathcal{L}_{\ell }(u)]-\frac{\Cov(\mathbb{\mathcal{L}_{\ell }}(u),
\widehat{C}_{\ell })}{\Var(\widehat{C}_{\ell })}\left (\widehat{C}_{\ell }-\mathbb{E}[\widehat{C}_{\ell }] \right ),
\end{equation}
which we again normalize by taking%
\[
\widetilde{\mathcal{L}}_{\ell |\widehat{C}_{\ell }}(u):=\frac{\mathcal{L}%
_{\ell |\widehat{C}_{\ell }}(u)}{\sqrt{\Var\left( \mathcal{L}_{\ell |%
\widehat{C}_{\ell }}(u)\right) }}.
\]%
We have then the next, fully degenerate, convergence result.

\begin{corollary}
As $\ell \rightarrow \infty ,$ for all $u\in \mathbb{R}$
\[
\left( \widetilde{\mathcal{L}}_{\ell |\widehat{C}_{\ell }}(u),\widetilde{%
\mathcal{L}}_{\ell |\widehat{C}_{\ell }}(0)\right) \rightarrow _{d}(Z,Z)%
\text{ ,}
\]%
where $Z$ denotes a standard Gaussian variable.
\end{corollary}
In short, boundary (non-zero level) and nodal length are asymptotically
independent, meaning also that the nodal length carries no information on
other functionals such as the excursion area, the Euler-Poincar\'{e}
characteristic or the number of critical points above a given threshold (see Remark \ref{parseval}).
This result, however, must be interpreted with great care: it is due to the
dominant role played by the sample norm in the behaviour of excursion sets.
When this effect is properly subtracted, the behaviour of length
fluctuations at any level are fully explained by the nodal length, in the
high-energy limit, and the joint distributions are completely degenerate.
Thus indeed the nodal lengths are asymptotically sufficient (in the high-energy limit) to characterize the measure of
the boundary at any threshold level, provided that the effect of random
fluctuations in the norm are properly taken into account. We refer to \cite{Fantaye} for some numerical evidence on these and related issues.


\subsection*{Acknowledgements}
We are grateful to Igor Wigman for suggesting this question and related topics for research, and for many insightful conversations. The research by DM was supported by the MIUR Excellence Department Project awarded to the Department of Mathematics, University of Rome Tor Vergata, CUP E83C18000100006; the research of MR was supported by the FSMP, the ANR-17-CE40-0008 project \emph{Unirandom} and the PRA 2018 49 project at the University of Pisa.

\section{Outline of the paper and proof of the main result}\label{Soutline}

The main ingredient behind our proofs is a neat series
representation (chaotic decomposition) of the length and a consequent careful investigation of its chaotic components. In particular, as briefly anticipated, the results recalled in \S\ref{Snodal} and \S \ref{Sboundary} can be interpreted in terms of Wiener-It\^o theory. Let us first recall the notion of Wiener chaos, restricting ourselves to our specific setting on the sphere (see  \cite[\S 2.2]{noupebook} and the references therein for a complete discussion).

\subsection{Wiener chaos}\label{subWiener}

Let us consider the sequence $\{H_k\}_{k\in \mathbb N}$ of Hermite polynomials on $\mathbb{R}$; these are defined as follows: $H_0 \equiv 1$ and
 $$H_{k}(t) := (-1)^k \gamma^{-1}(t) \frac{d^k}{dt^k} \gamma(t), \qquad k\in \mathbb N^*,$$
 where  $\gamma$ denotes the standard Gaussian density on $\mathbb R$.
The family $\mathbb{H} := \{(k!)^{-1/2} H_k, k\ge 0\}$ is a complete orthonormal system in the space of functions $L^2(\mathbb R, \mathcal B(\mathbb R), \gamma(t)dt)=:L^2(\gamma)$.
Bearing in mind \paref{fell}, we define the space ${\mathcal X}$ to be the closure in $L^2(\mathbb{P})$ of all real finite linear combinations of random variables $\xi$  of the form $\xi = r \, a_{\ell,0}$ for some $r\in \mathbb R$ or for $m\neq 0$  $$\xi = z \, a_{\ell,m} + \overline{z} \, (-1)^m a_{\ell, -m},\qquad z\in \mathbb{C},$$ thus ${\mathcal X}$ is a real centered Gaussian Hilbert subspace
of $L^2(\mathbb{P})$.
Let $q\ge 0$ be  an integer; we define the $q$-th Wiener chaos $C_q$ associated with ${\mathcal X}$ as the closure in $L^2(\mathbb{P})$ of all real finite linear combinations of random variables of the type
$$
H_{p_1}(\xi_1)\cdot H_{p_2}(\xi_2)\cdots H_{p_k}(\xi_k),\qquad k\ge 1,
$$
where  $p_1,...,p_k \in \mathbb N$ are such that $p_1+\cdots+p_k = q$, and $(\xi_1,...,\xi_k)$ is a standard real Gaussian vector extracted
from ${\mathcal X}$ (in particular, $C_0 = \mathbb{R}$).

The orthonormality and completeness of $\mathbb{H}$ in $L^2(\gamma)$, together with a standard monotone class argument \cite[Theorem 2.2.4]{noupebook}, implies that $C_q \,\bot\, C_m$ in  $L^2(\mathbb{P})$ for every $q\neq m$, and
\begin{equation*}
L^2(\Omega, \sigma({\mathcal X}), \mathbb{P}) = \bigoplus_{q=0}^\infty C_q,
\end{equation*}
that is, every square integrable real-valued functional $F$ of ${\mathcal X}$ can be (uniquely) represented as a series, converging in $L^2(\mathbb P)$, of the form
\begin{equation}\label{e:chaos2}
F 
= \sum_{q=0}^\infty \text{proj}[F| q],
\end{equation}
where $\text{proj}[F| q]:={\rm proj}(F \, | \, C_q)$ stands for the the projection of $F$ onto $C_q$ ($\text{proj}[F| 0]={\rm proj}(F \, | \, C_0) = \E [F]$ since $C_0=\mathbb R$).

\subsection{Chaotic expansions for lengths}\label{Swiener}

The perimeter of the boundary of excursion sets on the sphere can be (formally) written as
\begin{equation}\label{rapSphere}
\mathcal L_\ell(u) = \int_{\mathbb S^2} \delta_u(f_\ell(x)) \|\nabla f_\ell(x)\|\,dx
\end{equation}
where $\delta_u$ denotes the Dirac mass in $u$, $\nabla f_\ell$ the gradient field and $\| \cdot \|$ the Euclidean norm in $\R^2$.
Indeed, let us consider the $\eps$-approximating random variable ($\eps >0$)
$$
\mathcal L_\ell^\eps(u) := \frac{1}{2\eps} \int_{\mathbb S^2} 1_{[u-\eps, u+\eps]}(f_\ell(x))\, \| \nabla f_\ell(x)\|\, dx,
$$
where $1_{[u-\eps, u+\eps]}$ denotes the indicator function of the interval $[u-\eps, u+\eps]$. We have the following result whose proof is similar to the one given in the nodal case (for details see \cite[Appendix B]{MRW}), and hence omitted.
\begin{lemma}\label{asconv}
For $u\in \mathbb R$,
\begin{equation*}
\lim_{\eps \to 0} \mathcal L_\ell^\varepsilon(u) = \text{length}(f_\ell^{-1}(u)),
\end{equation*}
both a.s. and in $L^2(\P)$.
\end{lemma}
Lemma \ref{asconv} justifies \paref{rapSphere}. By a differentiation of  \eqref{fell} of $f_\ell$ it is easy to see that the random fields $f_\ell$ and the components of $\nabla f_\ell$, viewed as collections of
Gaussian random variables
indexed by $x\in\mathbb S^2$, are all lying in ${\mathcal X}$, hence $\mathcal L^\varepsilon_\ell(u), \mathcal L_\ell(u)\in L^2(\Omega, \sigma({\mathcal X}), \mathbb{P})$.

From the chaotic expansion of $\mathcal L_\ell^\eps(u)$ it is easy to obtain those of $\mathcal L_\ell(u)$ by letting $\eps$ go to zero. In order to recall the chaotic expansion (\ref{e:chaos2})
\begin{equation}\label{chaos_decomp}
\mathcal L^\eps_\ell(u) = \sum_{q=0}^{+\infty}  \text{proj}[\mathcal L_\ell^\eps(u) | {q}]
\end{equation}
for $\mathcal L_\ell^\eps(u)$,
let us first write
\begin{equation}\label{rapSphere2}
\mathcal L_\ell^\eps(u) =\frac{1}{2\eps} \sqrt{\frac{\ell(\ell+1)}{2}}\int_{\mathbb S^2} 1_{[u-\eps,u+\eps]}(f_\ell(x)) \|\widetilde \nabla f_\ell(x)\|\,dx,
\end{equation}
where  $\widetilde \nabla$ is the normalized gradient, i.e. $\widetilde \nabla:= \nabla / \sqrt{\frac{2}{\ell(\ell +1)}}$. (The variance of each component of $\nabla f_\ell(x)$ is $\ell(\ell +1)/2$.) Note that for each $x\in \mathbb S^2$, the random variables $f_\ell(x), \nabla f_\ell(x)$ are independent, and the components of $\nabla f_\ell(x)$ are independent as well.

In \cite[Proposition 7.2.2]{ROSSI2015} -- see \cite{GAFA2016, MRW} for the nodal case -- the terms of the series on the r.h.s. of \paref{chaos_decomp}
are explicitly given.
Let us introduce now two sequences of real numbers $\lbrace \beta^\eps_{k}(u)\rbrace_{k=0}^{+\infty}$ and $\lbrace \alpha_{2n,2m}\rbrace_{n,m=0}^{+\infty}$ corresponding to the chaotic coefficients of the indicator function $1_{[u-\eps,u+\eps]}(\cdot)/(2\eps)$ and the Euclidean norm respectively: for  $k=0, 1, 2, \dots$,
\begin{equation}\label{beta}
\begin{split}
\beta^\eps_{k}(u) &:= \frac{1}{2\eps} \int_{u-\eps}^{u+\eps} H_k(t)\gamma(t)\,dt \ \mathop{\longrightarrow}^{\eps\to 0}\  \gamma(u) H_k(u) =: \beta_k(u),
\end{split}
\end{equation}
where $\gamma$ still denotes the standard Gaussian density,
while for $n,m=0, 1, 2, \dots$
\begin{equation}\label{alpha}
\alpha _{2n,2m}:=\sqrt{\frac{\pi }{2}}\frac{(2n)!(2m)!}{n!m!}\frac{1}{2^{n+m}%
}p_{n+m}\left(\frac{1}{4}\right),
\end{equation}%
 $p_{N}$ denoting the swinging factorial coefficient
\begin{equation*}
p_{N}(x):=\sum_{j=0}^{N}(-1)^{j}(-1)^{N}\left(
\begin{array}{c}
N \\
j%
\end{array}%
\right) \frac{(2j+1)!}{(j!)^{2}}x^{j}.
\end{equation*}
Note that, if $u=0$, the coefficient $\beta^\eps_{k}(0)=0$ (and hence $\beta_k(0)=0$) whenever $k$ is odd.
The proof of the following is analogous to the one given in the nodal case (see \cite[\S 2]{MRW} and \cite[Proposition 7.2.2]{ROSSI2015}) and hence omitted.
\begin{proposition}\label{prop_exp}
The chaotic expansion \paref{e:chaos2} of the approximate length is
\begin{equation}\label{chaos_decomp_sphere_eps}
\begin{split}
\mathcal{L}_{\ell }^\eps(u) =\sum_{q=0}^{+\infty}  \text{proj}[\mathcal L_\ell^\eps(u)| {q}] & =   \sum_{q=0}^{\infty }\sqrt{\frac{\ell (\ell
+1)}{2}}\sum_{2a + 2b + c = q}\frac{\alpha
_{2a,2b}\beta _{c}^\eps(u)}{(2a)!(2b)!c!}\times \\
&\times \int_{\mathbb{S}^{2}}H_{c}(f_{\ell }(x))H_{2a}(\widetilde \partial
_{1;x}f_{\ell }(x))H_{2b}(\widetilde \partial
_{2;x}f_{\ell }(x))\,dx,
\end{split}
\end{equation}
where we use spherical coordinates (colatitude $\theta ,$ longitude $%
\varphi $) and for $x=(\theta _{x},\varphi _{x})$ we are using the notation
\begin{equation*}
\widetilde \partial _{1;x}=\left. (\ell(\ell+1)/2)^{-1/2}\cdot \frac{\partial }{\partial \theta }\right\vert
_{\theta =\theta _{x}},\quad \widetilde \partial _{2;x}= (\ell(\ell+1)/2)^{-1/2}\cdot\left. \frac{1}{\sin \theta }%
\frac{\partial }{\partial \varphi }\right\vert _{\theta =\theta _{x},\varphi,
=\varphi _{x}},
\end{equation*}
and the chaotic coefficients are as in \paref{beta} and \paref{alpha}. By letting $\eps \to 0$ we find
\begin{equation}\label{chaos_decomp_sphere}
\begin{split}
\mathcal{L}_{\ell }(u) =\sum_{q=0}^{+\infty}  \text{proj}[\mathcal L_\ell(u)| {q}] & =   \sum_{q=0}^{\infty }\sqrt{\frac{\ell (\ell
+1)}{2}}\sum_{2a + 2b + c = q}\frac{\alpha
_{2a,2b}\beta _{c}(u)}{(2a)!(2b)!c!}\times \\
&\times \int_{\mathbb{S}^{2}}H_{c}(f_{\ell }(x))H_{2a}(\widetilde \partial
_{1;x}f_{\ell }(x))H_{2b}(\widetilde \partial
_{2;x}f_{\ell }(x))\,dx.
\end{split}
\end{equation}
\end{proposition}
As it will be clearer later, it suffices to deal with the first few terms of the series in \paref{chaos_decomp_sphere}. For every $u\in \mathbb R$,
\begin{equation}\label{chaos0}
\text{proj}[\mathcal{L}_{\ell }(u)|0]=\gamma (u)\sqrt{ \frac{\lambda _{\ell }}{2}}
4\pi\sqrt{\frac{\pi }{2}} = \mathbb E[\mathcal L_\ell(u)],
\end{equation}
cf. \paref{mediau},
\begin{equation}\label{chaos1}
\text{proj}[\mathcal{L}_{\ell }(u)|1] = 0
\end{equation}
since spherical harmonics have zero mean on the sphere, and thanks to Green formula
(see \cite[Proposition 7.3.1]{ROSSI2015})
\begin{equation}\label{chaos2}
\text{proj}[\mathcal{L}_{\ell }(u)|2]=\sqrt{ \frac{\lambda _{\ell }}{2}}\sqrt{\frac{\pi }{2}} u^{2}\phi (u) \frac{1}{2}\int_{%
\mathbb{S}^{2}}H_{2}(f_{\ell }(x))\,dx,
\end{equation}
cf. \paref{D}; notice that the right-hand side of \paref{chaos2} is identically equal to zero if and only if $u=0$ (see also \cite{MRW}).
Obviously, if $u=0$, the term $\text{proj}[\mathcal{L}_{\ell }(0)|q]$ vanishes whenever $q$ is odd.

\smallskip

We are in a position to make some more comments, comparing results recalled in \S\ref{Snodal} and \S\ref{Sboundary}, and the theory of chaotic decompositions exploited above: the leading term in the $L^{2}(\mathbb P)$-expansion of the boundary length
around its expected value is provided by its orthogonal projection on the second-order
Wiener chaos \paref{chaos2}, rather than the fourth as in the nodal case (see \paref{M} and \paref{mrw} and compare with Proposition \ref{propdom});
as a consequence, the asymptotic variance is of order $\ell ,$ rather
than $\log \ell$ (Berry's cancellation phenomenon). Similarly to the nodal length case \paref{M}, the projection \paref{chaos2} takes a very simple
form, as the integral depends only on the (second power of the) random
eigenfunction, and not on its gradient.

\subsection{Proof of the main result}

In this paper we complete the characterization of the chaos expansion for
the boundary length of excursion sets, and indeed we show the following results which are of independent interest.

\begin{proposition}\label{varu4}
For all $u\in \mathbb{R},$ as $\ell \rightarrow \infty $
$$
\Var\left (\sum_{q=3}^{+\infty} \text{proj}[\mathcal L_\ell(u)| q] \right ) = \frac{\pi}{4} \phi(u)^2 \left ( H_{4}(u)+2H_{2}(u)-\frac{3}{2}\right )^2  \log \ell + O(1).
$$
\end{proposition}
Note that for $u=0$, from Proposition (\ref{prop_exp}), Proposition \ref{varu4} gives an alternative proof for \paref{varnodal}. Let us set for notational simplicity
$$
\mathcal M_\ell(u) := \sqrt{ \frac{\lambda _{\ell }}{2}}\sqrt{\frac{\pi }{2}}\phi (u)\left ( H_{4}(u)+2H_{2}(u)-\frac{3}{2}%
\right ) \frac{1}{4!}\int_{\mathbb{S}^{2}}H_{4}(f_{\ell }(x))dx.
$$
($\mathcal M_\ell(0) = \mathcal M_\ell$ in (\ref{M}).)
It is easy to show that, as $\ell\to +\infty$,
\begin{equation}\label{varMu}
\Var \left( \mathcal M_\ell(u) \right ) = \frac{\pi}{4} \phi(u)^2 \left ( H_{4}(u)+2H_{2}(u)-\frac{3}{2}\right )^2  \log \ell + O\left ( 1 \right ).
\end{equation}
\begin{proposition}\label{propdom}
\label{4thchaos} For all $u\in \mathbb{R},$ as $\ell \rightarrow \infty $%
$$
\Corr\left (\sum_{q=3}^{+\infty} \text{proj}[\mathcal L_\ell(u)| q], \mathcal M_\ell(u)  \right ) \to 1.
$$
\end{proposition}
Note that Proposition \ref{propdom} generalizes results in \cite{MRW}, moreover implies that
\begin{equation}\label{cov4}
\Cov(\text{proj}[\mathcal L_\ell(u)|4], \mathcal M_\ell(u)) \sim \frac{\pi}{4} \phi(u)^2 \left ( H_{4}(u)+2H_{2}(u)-\frac{3}{2}\right )^2  \log \ell.
\end{equation}
Thanks to Proposition \ref{varu4}, \paref{varMu} and \paref{cov4} we have that, as $\ell \to +\infty$,
\begin{equation}
\Var\left ( \text{proj}[\mathcal L_\ell(u)|4]  \right ) \sim \frac{\pi}{4} \phi(u)^2 \left ( H_{4}(u)+2H_{2}(u)-\frac{3}{2}\right )^2  \log \ell,
\end{equation}
so the contribution of $\text{proj}[\mathcal L_\ell(u)|3] + \sum_{q=5}^{+\infty} \text{proj}[\mathcal L_\ell(u)| q] $ is negligible, moreover we get the following.
\begin{corollary}
For all $u\in \mathbb{R},$ as $\ell \rightarrow \infty $%
$$
\Corr\left (\text{proj}[\mathcal L_\ell(u)| 4], \mathcal M_\ell(u)  \right ) \to 1.
$$
\end{corollary}
The proofs of Proposition \ref{varu4} and Proposition \ref{propdom} are delicate and technical and will be given in the next sections. We are now ready to prove our main result.

\begin{proof}[Proof of Theorem \ref{mainth}]
Let us consider $u_1, u_2\in \mathbb R$; recall that
\begin{equation*}
\rho({\mathcal L}_\ell(u_1), {\mathcal L}_\ell(u_2)) = \frac{\Cov(\mathcal L_\ell(u_1), \mathcal L_\ell(u_2))}{\sqrt{\Var(\mathcal L_\ell(u_1))\Var(\mathcal L_\ell(u_2))}}.
\end{equation*}
Assume first that $u_1 \cdot u_2 \ne 0$, 
thanks to \paref{lengthu}
\begin{equation}\label{bellou}
\mathcal L_\ell(u_i) - \E[\mathcal L_\ell(u_i) ]= \text{proj}[\mathcal L_\ell(u_i)|2] + \sum_{q=3}^{+\infty} \text{proj}[\mathcal L_\ell(u_i)|q]=\mathcal D_\ell(u_i) + \mathcal R_\ell(u_i),
\end{equation}
where $\Var(\mathcal R_\ell(u_i)) = o(\ell)$, as $\ell \to +\infty$.
Let us bear in mind \paref{varu}, \paref{varDell} and \paref{bellou}, thanks to Cauchy-Schwartz inequality we find, as $\ell \to +\infty$,
\begin{equation}\label{A1}
\begin{split}
\rho({\mathcal L}_\ell(u_1), {\mathcal L}_\ell(u_2)) \to 1.
\end{split}
\end{equation}
 For the nodal length we have, from  \paref{mrw},
 \begin{equation}\label{bello0}
 \mathcal L_\ell(0) - \E[\mathcal L_\ell(0) ]= \text{proj}[\mathcal L_\ell(0)|4] + \sum_{q=3}^{+\infty} \text{proj}[\mathcal L_\ell(0)|2q]=\mathcal M_\ell + \mathcal R_\ell,
\end{equation}
where $\mathcal M_\ell$ is as in \paref{M} and $\Var(\mathcal R_\ell) = O(1)$ as $\ell\to +\infty$. From \paref{bello0} and \paref{bellou} we have, for $u_1\ne 0$ and as $\ell\to +\infty$,
$$
\rho({\mathcal L}_\ell(u_1), {\mathcal L}_\ell(0)) \to 0,
$$
which together with \paref{A1} gives \paref{eqteo1}.

Recall now \paref{partialcorr} and \paref{explanatory}, here $Z=\| f_\ell\|^2_{L^2(\mathbb S^2)}$. From Proposition \ref{prop_exp}, 
we have, for every $u\in \mathbb R$ and as $\ell\to +\infty$,
\begin{equation}\label{resu}
\begin{split}
\mathcal L_\ell(u)^* =& \sum_{q=3}^{+\infty} \text{proj}[\mathcal L_\ell(u)|q].
\end{split}
\end{equation}
From Proposition \ref{varu4} and Proposition \ref{propdom} we have, for $u_1,u_2 \in \mathbb R$,
$$\lim_{\ell\to +\infty}\rho_{\| f_\ell\|^2_{L^2(\mathbb S^2)}}(\mathcal L_\ell(u_1), \mathcal L_\ell(u_2)) = 1$$
which proves \paref{eqteo2}. The proof of Theorem \ref{mainth} is hence complete.
\end{proof}

\subsection{Discussion}\label{Sdiscussion}

It can be instructive to compare the results in this paper with other recent
characterizations which have been given for the asymptotic distribution for
the nodal length of random eigenfunctions in the non-spherical case. We
recall first that a (non-universal) non-central limit theorem for arithmetic random waves,
i.e. Gaussian Laplacian eigenfunctions on the standard two-dimensional flat torus $\mathbb{T}^{2}:=\mathbb R^2/\mathbb Z^2,$
was established in \cite{GAFA2016}. To obtain this result, analogously to
our discussion above the nodal length was decomposed into
chaotic components (see \S \ref{Swiener}).
%
%
The expansion of nodal length in the toroidal
and spherical cases have both analogies and important differences. In both
cases, the term corresponding to $q=2$ disappears at $u=0,$ thus entailing
that the variance becomes of lower order (the so-called Berry's cancellation
phenomenon). Likewise, in both cases the nodal length is dominated by the
fourth-order chaos: however, it is only in the spherical case that the
fourth-order term admits an expression depending on the field only (and not
on the gradient components). Because of this, we do not expect that taking
into-account the random norm behaviour will be enough to establish full
correlation between nodal length and boundary curves
(it could be the case
that a degeneracy occurs when a sufficient number of different levels is
considered).

Similar cancellation phenomena occur for other geometric functionals,
including the excursion area and the Defect (\cite{DI}, \cite{MW2014}, \cite%
{MR2015}, \cite{ROSSI2016}, which cover any dimension $d\geq 2$), the Euler-Poincar\'{e}
characteristic \cite{CaMar2016}, and the zeros of complex arithmetic
random waves \cite{DNPR2016}; quantitative central limit theorems have
been given on the sphere in \cite{CaMar2016, DI,MW2014,MR2015} for
many of these statistics, in the high-energy limit where $\ell \rightarrow
\infty .$ On the torus the asymptotic behaviour has been shown to be more
complicated, because it is nonGaussian and differs across different
subsequences as the eigenvalue diverges (see \cite{GAFA2016}, \cite{DNPR2016}).

\section{Proof of Proposition \ref{varu4}}

Let us bear in mind Proposition \ref{prop_exp}. We need to take care first of second chaotic components; thanks to Green's formula (for details see \cite{ROSSI2015}) we can write
\begin{eqnarray}\label{green2}
\text{proj}[\mathcal L_\ell^\varepsilon(u)|2] =\sqrt{\frac{\ell(\ell+1)}{2} } \left ( \frac12 \beta^\varepsilon_{2}(u) \alpha_{0,0} + \beta^\varepsilon_0(u) \alpha_{2,0}   \right ) \int_{\mathbb S^2} H_2(f_\ell(x))\,dx
\end{eqnarray}
(note that $\text{proj}[\mathcal L_\ell(u) |2] = \mathcal D_\ell(u)$ in \paref{D}).
Let us hence set
\begin{eqnarray}\label{psi}
\mathcal L_\ell^\varepsilon(u) &=:& \int_{\mathbb S^2} \Psi_\ell^\varepsilon(x)\,dx, \qquad
\text{proj}[\mathcal L_\ell^\varepsilon(u) | q] =: \int_{\mathbb S^2} \Psi_\ell^\varepsilon(x;q)\,dx;\cr
\mathcal L_\ell(u) &=:& \int_{\mathbb S^2} \Psi_\ell(x)\,dx, \qquad
\text{proj}[\mathcal L_\ell(u) | q] =: \int_{\mathbb S^2} \Psi_\ell(x;q)\,dx,
\end{eqnarray}
where in particular for $q=2$ we refer to \paref{D} and \paref{green2}.
Before we proceed we need to introduce some more notation: let us fix $\overline{x} = (0,0)$ to be the ``north pole" and $y(\theta) =(0,\theta)$ to be points on the meridian where $\varphi =0$.
We will split the proof of Proposition \ref{varu4} into some lemmas.
\begin{lemma}\label{lemmaFubini}
For $C>0$
\begin{equation}\label{ris1}
\begin{split}
\Var\left (\sum_{q=3}^{+\infty} \text{proj}[\mathcal L_\ell(u) | q]  \right ) =  O(1) + 2\cdot 8\pi^2 \sum_{q=3}^{+\infty} \int_{C/\ell}^{\pi/2}  \E\left [  \Psi_\ell(\overline x;q) \Psi_\ell(y(\theta);q) \right ]\sin \theta\, d\theta,
\end{split}
\end{equation}
where the constant involved in the $O$-notation does not depend on $\ell$.
\end{lemma}
Let us deal with the terms of the series on the r.h.s. of \paref{ris1}.
\begin{lemma}\label{lemma3}
We have
\begin{eqnarray}\label{3O}
 \int_{C/\ell}^{\pi/2} \E\left [  \Psi_\ell(\overline x;3) \Psi_\ell(y(\theta);3) \right ]\sin \theta\,d\theta &=& O(1),\qquad \text{as } \ell\to +\infty.
\end{eqnarray}
\end{lemma}
\begin{lemma}\label{lemma4}
We have
\begin{eqnarray}\label{4O}
2\cdot 8\pi^2 \int_{C/\ell}^{\pi/2} \E\left [  \Psi_\ell(\overline x;4) \Psi_\ell(y(\theta);4) \right ]\sin \theta d\theta = &&\frac{\pi}{4} \phi(u)^2 \left ( H_{4}(u)+2H_{2}(u)-\frac{3}{2}\right )^2  \log \ell \cr
&&+ O_{\ell\to+\infty}(1).
\end{eqnarray}
\end{lemma}
\begin{lemma}\label{lemma5}
We have
\begin{eqnarray}\label{5O}
 \sum_{q=5}^{+\infty} \int_{C/\ell}^{\pi/2} \E\left [  \Psi_\ell(\overline x;q) \Psi_\ell(y(\theta);q) \right ] \sin \theta\,d\theta=O_{\ell\to+\infty}(1).
\end{eqnarray}
\end{lemma}
The proofs of Lemmas \ref{lemma3}-\ref{lemma5} are technical and will be given just below, the proof of Lemma \ref{lemmaFubini} is postponed to the Appendix. Let us now prove Proposition \ref{varu4}.

\begin{proof}[Proof of Proposition \ref{varu4}]
From \paref{ris1} we have,  thanks to Lemma \ref{lemma3}-Lemma \ref{lemma5},
\begin{eqnarray*}
\Var\left (\sum_{q=3}^{+\infty} \text{proj}[\mathcal L_\ell(u) | q]  \right ) &=&  2\cdot 8\pi^2 \sum_{q=3}^{+\infty} \int_{C/\ell}^{\pi/2}  \E\left [  \Psi_\ell(\overline x;q) \Psi_\ell(y(\theta);q) \right ]\sin \theta\, d\theta + O(1)\cr
&=& \frac{\pi}{4} \phi(u)^2 \left ( H_{4}(u)+2H_{2}(u)-\frac{3}{2}\right )^2  \log \ell + O_{\ell\to+\infty}(1)
\end{eqnarray*}
that immediately concludes the proof.
\end{proof}

\subsection{Proof of Lemma \ref{lemma3}}\label{technical}

 We will need the following.
\begin{eqnarray}\label{momento3}
&&\frac{\ell (\ell +1)}{2} \int_{C/\ell}^{\pi/2} P_\ell(\cos \theta)^3 \sin \theta\, d\theta = O_{\ell\to +\infty}(1);\cr
&&\int_{C/\ell} ^{\pi/2} P_\ell(\cos \theta) (P'_\ell(\cos \theta) \sin \theta)^2 \sin \theta \,d\theta
= O_{\ell\to +\infty}(1);\cr
&&\frac{2}{\ell(\ell+1)} \int_{C/\ell} ^{\pi/2} P_\ell(\cos \theta) (P_\ell(\cos \theta)\cos \theta - P_\ell^{''}(\cos \theta)\sin^2\theta)^2\sin\theta\,d\theta =  O_{\ell\to +\infty}(1);\cr
&&\frac{2}{\ell(\ell+1)} \int_{C/\ell} ^{\pi/2}(P_\ell(\cos \theta) \sin \theta)^2 (P_\ell(\cos \theta)\cos \theta - P_\ell^{''}(\cos \theta)\sin^2\theta) \,\sin\theta\,d\theta =  O_{\ell\to +\infty}(1);\cr
&& \frac{2}{\ell(\ell+1)} \int_{C/\ell} ^{\pi/2} P_\ell(\cos \theta)(P'_\ell(\cos \theta))^2 \sin\theta\,d\theta =  O_{\ell\to +\infty}(1).
\end{eqnarray}
To prove \paref{momento3} we will use some properties of Legendre polynomials and their derivatives recalled e.g. in \cite[Appendix A]{MRW}. Let us set $L:=\ell + \frac12$, and let $C>0$ be an absolute constant. We show the details of the analysis of the first and second estimate in \paref{momento3} since the other terms can be treated analogously.
\begin{eqnarray}\label{conti}
 &&\frac{\ell (\ell +1)}{2} \int_{C/\ell} ^{\pi/2} P_\ell(\cos \theta)^3 \sin \theta\, d\theta  =    \frac{\ell (\ell
+1)}{2} \frac{1}{L} \int_C^{L \pi/2} P_\ell(\cos \psi/L)^3 \sin \psi/L\,d\psi \cr
&&=    c \frac{\ell (\ell+1)}{2} \frac{1}{L} \int_C^{L \pi/2} \frac{1}{\ell^{3/2} \sin^{3/2} \psi/L} \left ( \sin (\psi + \pi/4) + O\left(\frac{1}/{\psi}\right) \right )^3 \sin \psi/L\,d\psi \cr
&& =   c \frac{\ell (\ell+1)}{2} \frac{1}{L} \int_C^{L \pi/2} \frac{(\sin \psi + \cos \psi)^3}{\ell^{3/2} \sin^{1/2} \psi/L} \,d\psi +  O\left ( \frac{\ell (\ell+1)}{2} \frac{1}{L} \int_C^{L \pi/2} \frac{1}{\ell} \frac{1}{\psi \sqrt \psi}\, d\psi \right )\cr
&& =   c' \frac{\ell (\ell+1)}{2} \frac{1}{L} \int_C^{L \pi/2} \frac{(\sin \psi + \cos \psi)^3}{\ell^{3/2} \sin^{1/2} \psi/L} \,d\psi + O(1),
\end{eqnarray}
for some constants $c, c' >0$.
Since
$$
\int_1^{+\infty} \frac{(\sin \psi + \cos \psi)^3}{\sqrt \psi}\,d \psi < +\infty
$$
from \paref{conti} we deduce
$$
\frac{\ell (\ell +1)}{2} \int_{C/\ell} ^{\pi/2} P_\ell(\cos \theta)^3 \sin \theta\, d\theta = O_{\ell\to +\infty}(1).
$$
To prove the second estimate in \paref{momento3}  we can write
\begin{eqnarray}\label{secondterm}
\int_{C/\ell} ^\pi P_\ell(\cos \theta) (P'_\ell(\cos \theta) \sin \theta)^2 \sin \theta \,d\theta
& = & \frac{1}{L} \int_C^{ L \pi/2}  P_\ell(\cos \psi/L)(P'_\ell(\cos \psi/L) \sin \psi/L)^2 \sin \psi/L \,d\psi\cr
& =&   \frac{1}{L} \int_C^{ L \pi/2}  \sqrt{\frac{2}{\ell \pi \sin \psi/L}} (\sin (\psi + \pi/4) + O\left(\frac{1}{\psi}\right) \times \cr
& \times &\frac{2}{\pi \ell \sin^3 \psi/L} (\ell \sin(\psi - \pi/4) + O(1))^2\sin^3\psi/L\,d\psi \cr
& = &   c  \int_C^{ L \pi/2}  \sqrt{\frac{1}{\ell \sin \psi/L}}\times \cr
&&\times  \sin (\psi + \pi/4)  \sin(\psi - \pi/4)^2 \,d\psi + O(1).
\end{eqnarray}
Since
$$
 \int_1^{ +\infty}  \frac{\sin (\psi + \pi/4)    \sin(\psi - \pi/4)^2}{\sqrt \psi} \,d\psi < +\infty
 $$
 we deduce from \paref{secondterm} that
 $$
 \int_{C/\ell} ^{\pi/2} P_\ell(\cos \theta) (P'_\ell(\cos \theta) \sin \theta)^2 \sin \theta \,d\theta = O_{\ell \to +\infty}(1).
 $$
 \begin{proof}[Proof of Lemma \ref{lemma3}]
We have
\begin{equation*}
\begin{split}
\E\left [  \Psi_\ell(\overline x;3) \Psi_\ell(y(\theta);3) \right ] =& \frac{\ell (\ell
+1)}{2}\mathbb E\Big [ \Big ( \frac{\alpha
_{0,0}\beta _{3}(u)}{3!} H_{3}(f_{\ell }(\overline{x})) \cr
&+ \frac{\alpha
_{2,0}\beta _{1}(u)}{2!} H_1(f_\ell(\overline x))(H_{2}(\widetilde \partial
_{1;x}f_{\ell }(\overline x)) + H_{2}(\widetilde \partial
_{2;x}f_{\ell }(\overline x))) \Big ) \times \cr
&\times \Big ( \frac{\alpha
_{0,0}\beta _{3}(u)}{3!} H_{3}(f_{\ell }(y(\theta))) \cr
&+ \frac{\alpha
_{2,0}\beta _{1}(u)}{2!} H_1(f_\ell(y(\theta)))(H_{2}(\widetilde \partial
_{1;x}f_{\ell }(y(\theta))) + H_{2}(\widetilde \partial
_{2;x}f_{\ell }(y(\theta)))) \Big )\Big ].
\end{split}
\end{equation*}
From a standard application of the Diagram Formula and \cite[(A1)-(A6)]{MRW}
\begin{eqnarray}\label{chaos3}
\E\left [  \Psi_\ell(\overline x;3) \Psi_\ell(y(\theta);3) \right ] =&& \frac{\ell (\ell
+1)}{2}  \Big ( a_1 (u) P_\ell(\cos \theta)^3 + a_2 (u) P_\ell(\cos \theta) \frac{2}{\ell(\ell+1)}(P'_\ell(\cos \theta) \sin \theta)^2  \cr
&&+a_3 (u) \frac{4}{\ell^2(\ell+1)^2} P_\ell(\cos \theta) (P_\ell(\cos \theta)\cos \theta - P_\ell^{''}(\cos \theta)\sin^2\theta)^2 \cr
&&+ a_4 (u) \frac{4}{\ell^2(\ell+1)^2} (P_\ell(\cos \theta) \sin \theta)^2 (P_\ell(\cos \theta)\cos \theta - P_\ell^{''}(\cos \theta)\sin^2\theta) \cr
&& + a_5 (u) P_\ell(\cos \theta) \frac{4}{\ell^2(\ell+1)^2}(P'_\ell(\cos \theta))^2
\Big ),
\end{eqnarray}
where $a_i(u)$, $i=1,2, \dots, 5$ are constants depending on $u$ (we do not need to compute them explicitly).
It remains to prove that, as $\ell \to +\infty$,
\begin{equation*}
\int_{C/\ell} ^{\pi/2} \E\left [  \Psi_\ell(\overline x;3) \Psi_\ell(y(\theta);3) \right ] = O(1)
\end{equation*}
that immediately follows plugging \paref{momento3} in \paref{chaos3}.
 \end{proof}

 \subsection{Proof of Lemma \ref{lemma4}}

We will need the following.
As $\ell \to +\infty$
\begin{eqnarray} \label{momento4}
&&\int_{C/\ell} ^{\pi/2} P_\ell(\cos \theta)^4 \sin \theta\,d\theta = \frac{4}{\pi^2} \frac38 \frac{\log \ell}{\ell^2} + O\left ( \frac{1}{\ell^2}\right); \cr
&&\int_{C/\ell} ^{\pi/2} P_\ell(\cos \theta)^2 \frac{2}{\ell(\ell+1)} (P'_\ell(\cos \theta) \sin \theta)^2 \sin \theta\,d\theta  =  \frac{1}{\pi^2} \frac{\log \ell}{\ell^2} + O\left ( \frac{1}{\ell^2}\right);\cr
&&\int_{C/\ell} ^{\pi/2} P_\ell^2(\cos \theta) \frac{4}{\ell^2(\ell+1)^2} (P_\ell(\cos \theta)\cos \theta - P_\ell^{''}(\cos \theta)\sin^2\theta)^2 \sin \theta\,d\theta = \frac{16}{\pi^2} \frac38 \frac{\log \ell}{\ell^2} + O\left ( \frac{1}{\ell^2}\right);\cr
&&\int_{C/\ell} ^{\pi/2} P_\ell(\cos \theta) \frac{4}{\ell^2(\ell+1)^2} (P'_\ell(\cos \theta) \sin \theta)^2 (P_\ell(\cos \theta)\cos \theta - P_\ell^{''}(\cos \theta)\sin^2\theta)\sin \theta\,d\theta \cr
&&= \frac{16}{\pi^2} \frac18 \frac{\log \ell}{\ell^2}+ O\left ( \frac{1}{\ell^2}\right);\cr
&&\int_{C/\ell} ^{\pi/2} \frac{4}{\ell^2(\ell+1)^2}P'_\ell(\cos \theta)^4  \sin^5 \theta\,d\theta
=  \frac{16}{\pi^2} \frac38 \frac{\log \ell}{\ell^2} + O\left ( \frac{1}{\ell^2}\right);\cr
&&\int_{C/\ell} ^{\pi/2} \frac{8}{\ell^3(\ell+1)^3}(P'_\ell(\cos \theta)\sin \theta)^2 (P_\ell(\cos \theta)\cos \theta - P_\ell^{''}(\cos \theta)\sin^2\theta)^2  \sin \theta\,d\theta
= \frac{32}{\pi^2} \frac18 \frac{\log \ell}{\ell^2} + O\left ( \frac{1}{\ell^2}\right);\cr
&&\int_{C/\ell} ^{\pi/2} \frac{2^4}{\ell^4(\ell+1)^4} (P_\ell(\cos \theta)\cos \theta - P_\ell^{''}(\cos \theta)\sin^2\theta)^4  \sin \theta\,d\theta
 =  \frac{4\cdot 2^4}{\pi^2} \frac38 \frac{\log \ell}{\ell^2} + O\left ( \frac{1}{\ell^2}\right);\cr
 && \int_{C/\ell} ^{\pi/2} P_\ell^2(\cos \theta) \frac{4}{\ell^2(\ell+1)^2} P'_\ell(\cos \theta)^2 \sin \theta\,d\theta = O\left ( \frac{1}{\ell^2}\right);\cr
&&  \int_{C/\ell} ^{\pi/2} \frac{2}{\ell(\ell+1)} (P'_\ell(\cos \theta) \sin \theta)^2 \frac{4}{\ell^2 (\ell+1)^2} P'_\ell(\cos \theta)^2\sin \theta\,d\theta = O\left ( \frac{1}{\ell^2}\right);\cr
&& \int_{C/\ell} ^{\pi/2} \frac{4}{\ell^2(\ell +1)^2}(P_\ell(\cos \theta)\cos \theta - P_\ell^{''}(\cos \theta)\sin^2\theta)^2
  \frac{4}{\ell^2 (\ell+1)^2} P'_\ell(\cos \theta)^2 \sin \theta\,d\theta = O\left ( \frac{1}{\ell^2}\right);\cr
  &&\int_{C/\ell} ^{\pi/2}  \frac{2^4}{\ell^4(\ell+1)^4} (P_\ell'(\cos \theta))^4\sin \theta\,d\theta = O\left ( \frac{1}{\ell^2}\right).
  \end{eqnarray}
To prove \paref{momento4} we will use again some properties of Legendre polynomials and their derivatives recalled e.g. in \cite[Appendix A]{MRW}. Remember that $L=\ell + \frac12$, and  $C>0$ is an absolute constant. We show the details of the analysis of the first and second estimate in \paref{momento4} since the other terms can be treated analogously.
We have
\begin{eqnarray}
\int_{C/\ell} ^{\pi/2} P_\ell(\cos \theta)^4 \sin \theta\,d\theta &=&  \frac{1}{L} \int_C^{L\pi/2} P_\ell(\cos \psi /L)^4 \sin \psi/L\,d\psi \cr
&=&  \frac{1}{L} \int_C^{L\pi/2} \left ( \frac{4\sin^4(\psi + \pi/4)}{\pi^2 \ell^2 \sin^2 \frac{\psi}{L} }  + O\left (\frac{1}{\psi^{3}} \right ) \right )   \sin \frac{\psi}{L}\,d\psi \cr
&=& O\left(\frac{1}{L^2}\right) +  \frac{4}{\pi^2} \frac{1}{L} \int_C^{L\pi/2}  \frac{\sin^4(\psi + \pi/4)}{\ell^2 \sin \frac{\psi}{L} } \,d\psi \cr
&=& O\left(\frac{1}{L^2}\right) +  \frac{4}{\pi^2} \frac{1}{L} \int_C^{L\pi/2}  \frac{(\frac38 -\frac18 \cos 4\psi +\frac12 \sin 2\psi)}{\ell^2 \sin \frac{\psi}{L} } \,d\psi \cr
&=& \frac{4}{\pi^2} \frac38 \frac{\log \ell}{\ell^2} + O\left(\frac{1}{L^2}\right).
\end{eqnarray}
In order to prove the second estimate in \paref{momento4} we write
\begin{eqnarray}
&&\int_{C/\ell} ^{\pi/2} P_\ell(\cos \theta)^2 \frac{2}{\ell(\ell+1)} (P'_\ell(\cos \theta) \sin \theta)^2 \sin \theta\,d\theta \cr
&&=   \frac{1}{L} \int_C^{L\pi/2} P_\ell(\cos \psi /L)^2 \frac{2}{\ell(\ell+1)} P'_\ell(\cos \psi/L)^2 \sin^3 \psi/L \,d\psi \cr
&=& \frac{1}{L} \int_C^{L\pi/2} \left ( \frac{2\sin^2(\psi + \pi/4)}{\pi \ell \sin \frac{\psi}{L} }  + O\left (\frac{1}{\psi^{2}} \right ) \right )  \times \cr
&& \times \frac{2}{\ell(\ell+1)} \left ( \frac{2\ell^2 \sin^2(\psi - \pi/4)}{\pi \ell \sin^3 \frac{\psi}{L} }  + O\left (\frac{\ell}{\psi^{3}} \right ) \right )\, \sin^3 \psi/L\,d\psi  \cr
&& = O\left(\frac{1}{L^2}\right) + \frac{1}{L}  \frac{2}{\ell(\ell+1)}\int_C^{L\pi/2} \frac{4\sin^2(\psi + \pi/4)\sin^2(\psi - \pi/4)}{\pi^2  \sin \frac{\psi}{L} }  \,d\psi  \cr
&& =  O\left(\frac{1}{L^2}\right) +  \frac{2}{\ell(\ell+1)}\int_C^{L\pi/2} \frac{4\cdot \frac18 (1 + \cos(4\psi))}{\pi^2 \psi }  \,d\psi  \cr
&& = O\left(\frac{1}{L^2}\right) + \frac{1}{\pi^2} \frac{\log \ell}{\ell^2}.
\end{eqnarray}

We will also need the following.
\begin{eqnarray}\label{calcolocoeff}
\alpha _{00}\beta _{4}(u) &=&\sqrt{\frac{\pi }{2}}\frac{1}{\sqrt{2\pi }}\exp
\left (-\frac{u^{2}}{2}\right )(u^{4}-6u^{2}+3)=\frac{1}{2}\exp \left (-\frac{u^{2}}{2}%
\right )(u^{4}-6u^{2}+3 ), \cr
\alpha _{20}\beta _{2}(u) &=&\frac{1}{2}\sqrt{\frac{\pi }{2}}\frac{1}{\sqrt{%
2\pi }}\exp \left (-\frac{u^{2}}{2}\right )(u^{2}-1 )=\frac{1}{4}\exp \left (-\frac{u^{2}}{2}%
\right )(u^{2}-1), \cr
\alpha _{40}\beta _{0}(u) &=&-\frac{3}{8}\sqrt{\frac{\pi }{2}}\frac{1}{\sqrt{%
2\pi }}\exp \left (-\frac{u^{2}}{2}\right )=-\frac{3}{16}\exp \left (-\frac{u^{2}}{2}\right ).
\end{eqnarray}%
\begin{proof}[Proof of Lemma \ref{lemma4}]
We have
\begin{eqnarray}
\E\left [  \Psi_\ell(\overline x;4) \Psi_\ell(y(\theta);4) \right ] =&& \frac{\ell (\ell
+1)}{2}\mathbb E\Big [ \Big ( \frac{\alpha
_{0,0}\beta _{4}(u)}{4!} H_{4}(f_{\ell }(\overline{x})) \cr
&&+ \frac{\alpha
_{2,0}\beta _{2}(u)}{2! 2!} H_2(f_\ell(\overline x))(H_{2}(\widetilde \partial
_{1;x}f_{\ell }(\overline x)) + H_{2}(\widetilde \partial
_{2;x}f_{\ell }(\overline x))) \cr
&& + \frac{\alpha
_{2,2}\beta _{0}(u)}{2! 2!} H_{2}(\widetilde \partial
_{1;x}f_{\ell }(\overline x)) H_{2}(\widetilde \partial
_{2;x}f_{\ell }(\overline x)) \cr
&&  + \frac{\beta_0(u) \alpha_{4,0}}{4!} ( H_4(\widetilde \partial
_{1;x}f_{\ell }(\overline x)) + H_4(\widetilde \partial
_{2;x}f_{\ell }(\overline x))   \Big ) \times \cr
&&\times \Big ( \frac{\alpha
_{0,0}\beta _{4}(u)}{4!} H_{4}(f_{\ell }(y(\theta))) \cr
&&+ \frac{\alpha
_{2,0}\beta _{2}(u)}{2! 2!} H_2(f_\ell(y(\theta)))(H_{2}(\widetilde \partial
_{1;x}f_{\ell }(y(\theta))) + H_{2}(\widetilde \partial
_{2;x}f_{\ell }(y(\theta)))) \cr
&& + \frac{\alpha
_{2,2}\beta _{0}(u)}{2! 2!} H_{2}(\widetilde \partial
_{1;x}f_{\ell }(\overline y(\theta))) H_{2}(\widetilde \partial
_{2;x}f_{\ell }(\overline y(\theta)))\cr
 &&  + \frac{\beta_0(u) \alpha_{4,0}}{4!} ( H_4(\widetilde \partial
_{1;x}f_{\ell }(y(\theta))) + H_4(\widetilde \partial
_{2;x}f_{\ell }(y(\theta)))  \Big )
\Big ].
\end{eqnarray}
Thanks to Diagram Formula and \cite[(A1)-(A6)]{MRW} we get
\begin{eqnarray}\label{conti4}
\mathbb E\left [  \Psi_\ell(\overline x;4) \Psi_\ell(y(\theta);4) \right ] =&& \frac{\ell (\ell+1)}{2} \Big [ \left( \frac{\alpha
_{0,0}\beta _{4}(u)}{4!}\right )^2 4! P_\ell(\cos \theta)^4  \cr
&&+ 2\cdot \frac{\alpha
_{0,0}\beta _{4}(u)}{4!} \frac{\alpha
_{2,0}\beta _{2}(u)}{2! 2!} 4! P_\ell^2(\cos \theta) \frac{2}{\ell(\ell+1)} (P'_\ell(\cos \theta) \sin \theta)^2 \cr
&& + \left(  \frac{\alpha
_{2,0}\beta _{2}(u)}{2! 2!}\right)^2 \Big ( 4 P_\ell^2(\cos \theta) \frac{4}{\ell^2(\ell+1)^2} (P_\ell(\cos \theta)\cos \theta - P_\ell^{''}(\cos \theta)\sin^2\theta)^2 \cr
&& - 16 P_\ell(\cos \theta) \frac{4}{\ell^2(\ell+1)^2} (P'_\ell(\cos \theta) \sin \theta)^2 (P_\ell(\cos \theta)\cos \theta - P_\ell^{''}(\cos \theta)\sin^2\theta) \cr
&& + 4 \frac{4}{\ell^2(\ell+1)^2}(P'_\ell(\cos \theta)\sin \theta)^4 \Big ) \cr
&& + 2\cdot  \frac{\beta_0(u)\alpha_{40}}{4!} \frac{\beta_4(u)\alpha_{00}}{4!} 4!  \frac{4}{\ell^2(\ell+1)^2}(P'_\ell(\cos \theta)^4 \sin^4 \theta)  \cr
&& + 2 \cdot \frac{\beta_0(u)\alpha_{40}}{4!} \frac{\beta_2(u) \alpha_{20}}{2! 2!} 4!  \frac{8}{\ell^3(\ell+1)^3}(P'_\ell(\cos \theta)\sin \theta)^2\times \cr
&& \times  (P_\ell(\cos \theta)\cos \theta - P_\ell^{''}(\cos \theta)\sin^2\theta)^2\cr
&&+ \left (\frac{\beta_0(u)\alpha_{40}}{4!} \right )^2 4! \frac{2^4}{\ell^4(\ell+1)^4} (P_\ell(\cos \theta)\cos \theta - P_\ell^{''}(\cos \theta)\sin^2\theta)^4 \cr
&& + \left(  \frac{\alpha
_{2,0}\beta _{2}(u)}{2! 2!}\right)^2 4 \cdot P_\ell^2(\cos \theta) \frac{4}{\ell^2(\ell+1)^2} P'_\ell(\cos \theta)^2 \cr
&& +  2 \cdot \frac{\alpha
_{2,0}\beta _{2}(u)}{2! 2!}  \frac{\alpha
_{2,2}\beta _{0}(u)}{2! 2!} 4 \cdot \frac{2}{\ell(\ell+1)} (P'_\ell(\cos \theta) \sin \theta)^2 \frac{4}{\ell^2 (\ell+1)^2} P'_\ell(\cos \theta)^2 \cr
&&+ \left( \frac{\alpha
_{2,2}\beta _{0}(u)}{2! 2!} \right )^2 4 \cdot \frac{4}{\ell^2(\ell +1)^2}(P_\ell(\cos \theta)\cos \theta - P_\ell^{''}(\cos \theta)\sin^2\theta)^2\times \cr
&& \times  \frac{4}{\ell^2 (\ell+1)^2} P'_\ell(\cos \theta)^2 + \left (\frac{\beta_0(u)\alpha_{40}}{4!} \right )^2 4!  \frac{2^4}{\ell^4(\ell+1)^4} (P_\ell'(\cos \theta))^4 \Big ].
\end{eqnarray}
Plugging \paref{momento4} and \paref{calcolocoeff} into \paref{conti4}  we get, as $\ell\to+\infty$,
\begin{eqnarray*}
2 \cdot 8\pi^2 \int_{C/\ell} ^{\pi/2} \mathbb E\left [  \Psi_\ell^\varepsilon(\overline x;4) \Psi_\ell^\varepsilon(y(\theta);4) \right ] \sin\theta \,d\theta =  \frac{\pi}{4} \phi(u)^2 \left ( H_{4}(u)+2H_{2}(u)-\frac{3}{2}\right )^2  \log \ell + O(1)
\end{eqnarray*}
that allows to conclude.
\end{proof}

\subsection{Proof of Lemma \ref{lemma5}}

This part is inspired by the proof of Lemma 3.5 in \cite{DNPR2016}. We will need the following.
\begin{lemma}\label{lemmaOgrande}
As $\ell\to+\infty$
\begin{eqnarray}\label{result}
 \max_{\substack{m_1,\dots, m_4 \ge 0\\ m_1+\dots + m_4 =5}} && \left (\frac{2}{\ell(\ell+1)} \right )^{\frac12 m_2 +m_3+m_4}  \int_{C/\ell}^{L\pi/2} |P_\ell(\cos \theta)|^{m_1} |P_\ell(\cos \theta)\sin \theta|^{m_2} \times \cr
&&\times  |P_\ell(\cos \theta)\cos \theta - P_\ell^{''}(\cos \theta)\sin^2\theta|^{m_3}   |P'_\ell (\cos \theta)|^{m_4}\sin \theta\,d\theta \cr
&&= O\left (\frac{1}{\ell^2}\right ).
\end{eqnarray}
\end{lemma}
\begin{proof}
Thanks to the estimates for Legendre polynomials and their derivatives recalled in \cite[Appendix A]{MRW} we have
\begin{eqnarray*}
&& \left (\frac{2}{\ell(\ell+1)} \right )^{\frac12 m_2 +m_3+m_4} \frac{1}{L}  \int_C^{L\pi}  |P_\ell(\cos \psi/L|^{m_1} |P_\ell(\cos \psi/L)\sin \psi/L|^{m_2} \times \cr
&& |P_\ell(\cos \psi/L)\cos \psi/L - P_\ell^{''}(\cos \psi/L)\sin^2\psi/L|^{m_3}   |P'_\ell (\cos \psi/L)|^{m_4}\sin \psi/L\,d\psi \cr
&& = O\left ( \frac{1}{L^2} \int_C^{L\pi/2} \psi^{-3/2 - m_4}\, d\psi \right ) = O\left (\frac{1}{\ell^2}\right )
\end{eqnarray*}
which gives \paref{result}.
\end{proof}

\begin{proof}[Proof of Lemma \ref{lemma5}]
As for the proof of Lemma \ref{lemmaFubini} in the Appendix
\begin{eqnarray}\label{conto}
&& \left | \int_{C/\ell}^{\pi/2} \E\left [  \Psi_\ell(\overline x;q) \Psi_\ell(y(\theta);q) \right ] \sin\theta\,d\theta\right | \cr
&&= \Big | \frac{\ell (\ell
+1)}{2}\sum_{2a + 2b + c = q}\sum_{2a' + 2b' + c'=q} \frac{\alpha
_{2a,2b}\beta _{c}(u)}{(2a)!(2b)!c!}\frac{\alpha
_{2a',2b'}\beta _{c'}(u)}{(2a')!(2b')!(c')!} \times \cr
&&\times \frac{1}{L}\int_{C/\ell}^{\pi/2} \mathbb E\Big[H_{c}(f_{\ell }(\overline x))H_{2a}(\widetilde \partial
_{1;x}f_{\ell }(\overline x))H_{2b}(\widetilde \partial
_{2;x}f_{\ell }(\overline x)) \times \cr
&&\times H_{c'}(f_{\ell }(y(\theta)))H_{2a}(\widetilde \partial
_{1;x}f_{\ell }(y(\theta))) H_{2b}(\widetilde \partial
_{2;x}f_{\ell }(y(\theta))) \Big ] \sin \theta\,d\theta \Big | \cr
&&\le   \ell (\ell
+1) \sum_{2a + 2b + c = q}\sum_{2a' + 2b' + c'=q} \left | \frac{\alpha
_{2a,2b}\beta_{c}(u)}{(2a)!(2b)!c!} \right | \left |\frac{\alpha
_{2a',2b'}\beta_{c'}(u)}{(2a')!(2b')!(c')!} \right |  \times\cr
&&\times \Big | \mathcal V(2a,2b,c,2a',2b',c')   \Big |,
\end{eqnarray}
where $\mathcal V(2a,2b,c,2a',2b',c') $ is the sum of no more than $q!$ terms of the form
\begin{eqnarray}\label{explicit}
&&\left ( \frac{2}{\ell(\ell+1)} \right )^{\frac12 m_2 +m_3 +m_4}\int_{C/\ell}^{\pi/2}  P_\ell(\cos \theta)^{m_1} (P_\ell(\cos \theta)\sin \psi/L)^{m_2} \times \cr
&& \times (P_\ell(\cos \theta)\cos \theta - P_\ell^{''}(\cos \theta)\sin^2\theta)^{m_3}  P'_\ell (\cos \theta)^{m_4}\sin\theta\,d\theta,
\end{eqnarray}
where $m_1, \dots, m_4 \ge 0$ and $m_1+m_2+m_3+m_4 = q$. Since $C/\ell <\theta < \pi/2$ we have that the absolute value of each factor of the integrand in \paref{explicit} is less than $1-\delta$ for some small $\delta >0$, see i.e., the expressions for $P_{\ell}$, $P'_{\ell}$, $P''_{\ell}$ which are proved in \cite{CMW}, Lemma B3 and reported in \cite{MRW}, Appendix A.   Since $q\ge 5$ we can write for \paref{conto}, taking into account \paref{explicit},
\begin{eqnarray*}
&& \left | \int_{C/\ell}^{\pi/2} \E\left [  \Psi_\ell(\overline x;q) \Psi_\ell(y(\theta);q) \right ] \sin\theta\,d\theta \right | \cr
&&\le   \ell (\ell
+1) \sum_{2a + 2b + c = q}\sum_{2a' + 2b' + c'=q} \left | \frac{\alpha
_{2a,2b}\beta_{c}(u)}{(2a)!(2b)!c!} \right | \left |\frac{\alpha
_{2a',2b'}\beta_{c'}(u)}{(2a')!(2b')!(c')!} \right |  \times\cr
&&\times q! (1-\delta)^{q-5} \max_{m_1+\dots+ m_4 =5} \left ( \frac{2}{\ell(\ell+1)} \right )^{\frac12 m_2 +m_3 +m_4}\int_{C/\ell}^{\pi/2} |P_\ell(\cos \theta)|^{m_1} \times \cr
&&\times |P_\ell(\cos \theta)\sin \theta |^{m_2}
 |P_\ell(\cos \theta)\cos \theta - P_\ell^{''}(\cos \theta)\sin^2\theta|^{m_3}
 |P'_\ell (\cos \theta)|^{m_4}\sin \theta\,d\theta.
\end{eqnarray*}
Therefore we have
\begin{eqnarray}\label{semifinal1}
&&\left | \sum_{q=5}^{+\infty} \int_{C/\ell}^{\pi/2} \E\left [  \Psi^\varepsilon_\ell(\overline x;q) \Psi^\varepsilon_\ell(y(\theta);q) \right ] \sin \theta\,d\theta \right |\cr
&& \le  \ell (\ell
+1) (1-\delta)^{-5} \max_{m_1+\dots + m_4 =5}  \left ( \frac{2}{\ell(\ell+1)} \right )^{\frac12 m_2 +m_3 +m_4}\frac{1}{L}\int_C^{L\pi/2} |P_\ell(\cos \theta)|^{m_1}  \times \cr
&& |P_\ell(\cos \theta)\sin \theta |^{m_2}
 |P_\ell(\cos \theta)\cos \theta - P_\ell^{''}(\cos \theta)\sin^2\theta|^{m_3}  |P'_\ell (\cos \theta)|^{m_4}\sin \theta\,d\theta \times \cr
&&\times \sum_{q=5}^{+\infty}  q! (1-\delta)^q \sum_{2a + 2b + c = q}\sum_{2a' + 2b' + c'=q} \left | \frac{\alpha
_{2a,2b}\beta_{c}(u)}{(2a)!(2b)!c!} \right | \left |\frac{\alpha
_{2a',2b'}\beta_{c'}(u)}{(2a')!(2b')!(c')!} \right |.
\end{eqnarray}
For the series on the r.h.s. of \paref{semifinal1}, repeating the same argument as in the proof of Lemma 3.5 in \cite{DNPR2016}, we get
\begin{eqnarray*}
&&\sum_{q=5}^{+\infty}  q! (1-\delta)^q \sum_{2a + 2b + c = q}\sum_{2a' + 2b' + c'=q} \left | \frac{\alpha
_{2a,2b}\beta_{c}(u)}{(2a)!(2b)!c!} \right | \left |\frac{\alpha
_{2a',2b'}\beta_{c'}(u)}{(2a')!(2b')!(c')!} \right |  \cr
&& \le \sum_{a,b,c,a',b',c'\ge 0} \frac{\alpha
_{2a,2b}^2\beta_{c}(u)^2}{(2a)!(2b)!c!} \frac{(2a+2b+c)!}{(2a)!(2b)!c!} \sqrt{1-\delta}^{2a+2b+c+2a'+2b'+c'}.
\end{eqnarray*}
Now note that the map
\begin{equation*}
(a,b,c)\mapsto \frac{\alpha
_{2a,2b}^2\beta_{c}(u)^2}{(2a)!(2b)!c!}
\end{equation*}
is  bounded.  Indeed, see \paref{alpha} and \paref{beta}, and recall that there exists $C>0$ s.t. for every $k\in \mathbb N$ and $u\in \mathbb R$
\begin{equation*}
|H_k(u)| \gamma(u) \le C \sqrt{k!}.
\end{equation*}
Finally Lemma \ref{lemmaOgrande} applied in \paref{semifinal1} allows to conclude the proof.
\end{proof}

\section{Proof of Proposition \ref{propdom}}

We will need the following key result (inspired by Proposition 3.1 in \cite{MRW}) whose proof is in the Appendix. Recall that $
L:= \ell + \frac12,
$
and set
\begin{equation*}
\mathcal J_\ell(\psi;4) := - \sqrt{ \frac{\lambda _{\ell }}{2}}\sqrt{\frac{\pi }{2}}\phi (u)\frac{1}{4!} \left ( H_{4}(u)+2H_{2}(u)-\frac{3}{2}%
\right )  \frac{8\pi ^{2}}{L}\E\left [\Psi_\ell(\overline x;4) H_{4}(f_{\ell
}(y(\psi/L)))\right].
\end{equation*}

\begin{lemma} \label{keylemma}
For any constant $C>0$, uniformly over $\ell$ we have, for $0<\psi<C$,
\begin{equation}\label{J1}
\mathcal J_\ell(\psi;4) = O(\ell),
\end{equation}
and, for $C<\psi< L\frac{\pi}{2}$,
\begin{eqnarray}\label{J2}
\mathcal J_\ell(\psi;4) = &&  \frac14 \frac{\pi}{2}\phi(u)^2 \left ( H_4(u) + 2H_2(u) -\frac38 \right )^2  \frac{1}{\psi \sin \psi/L} \cr
&& +  a(u) \frac{\cos 4\psi}{\psi \sin \psi/L} + b(u) \frac{\sin 2\psi}{\psi \sin \psi/L} \cr
&& + O\left(\frac{1}{(\psi^2 \sin^2 \psi/L)}\right)+O\left(\frac{1}{\ell} \frac{1}{\psi \sin ^{2}\frac{\psi }{L}}\right),
\end{eqnarray}
where $a(u)$ and $b(u)$ are two (explicit) constants that depend on $u$.
\end{lemma}
\smallskip
\begin{proof}[Proof of Proposition \ref{propdom}]
It suffices to prove that, as $\ell \to +\infty$,
\begin{equation}\label{cov1}
\Cov\left (  \sum_{q=3}^{+\infty} \text{proj}[\mathcal L_\ell(u)|q], \mathcal{M}_{\ell }(u)\right ) \sim \frac{\pi}{4} \phi(u)^2 \left ( H_{4}(u)+2H_{2}(u)-\frac{3}{2}\right )^2  \log \ell
\end{equation}
(cf. \eqref{varMu} and Proposition \ref{varu4}).
By continuity of the inner product in $L^{2}$ spaces, we have
\begin{equation*}
\Cov\left (  \sum_{q=3}^{+\infty} \text{proj}[\mathcal L_\ell(u)|q], \mathcal{M}_{\ell }(u)\right ) = \lim_{\varepsilon \rightarrow 0}\Cov\left (  \sum_{q=3}^{+\infty} \text{proj}[\mathcal L^\varepsilon_\ell(u)|q], \mathcal{M}_{\ell }(u)\right ).
\end{equation*}
Recall \paref{psi}.
Note that $ \sum_{q=3}^{+\infty}  \Psi^\varepsilon_\ell(\cdot ;q)$ and $H_{4}(f_{\ell }(\cdot ))$ are both in $%
L^{2}(\mathbb{S}^{2}\times \Omega )$ and they are isotropic, and thus
\begin{eqnarray*}
&&\Cov\left (  \sum_{q=3}^{+\infty} \text{proj}[\mathcal L^\varepsilon_\ell(u)|q], \mathcal{M}_{\ell }(u)\right )= \Cov\left( \sum_{q=3}^{+\infty} \int_{\mathbb S^2} \Psi^\varepsilon_\ell(x;q)\,dx,\mathcal{M}_{\ell }(u)\right )\cr
&=& - \sqrt{ \frac{\lambda _{\ell }}{2}}\sqrt{\frac{\pi }{2}}\phi (u)\left ( H_{4}(u)+2H_{2}(u)-\frac{3}{2}%
\right ) \frac{1}{4!}\E\left [ \int_{%
\mathbb{S}^{2}}\sum_{q=3}^{+\infty}  \Psi^\varepsilon_\ell(x;q)\,dx \int_{\mathbb{S}^{2}}H_{4}(f_{\ell
}(y))dy\right] \\
&=& - \sqrt{ \frac{\lambda _{\ell }}{2}}\sqrt{\frac{\pi }{2}}\phi (u)\left ( H_{4}(u)+2H_{2}(u)-\frac{3}{2}%
\right ) \frac{1}{4!}\int_{%
\mathbb{S}^{2}} \int_{\mathbb{S}^{2}}\E\left [ \sum_{q=3}^{+\infty}  \Psi^\varepsilon_\ell(x;q) H_{4}(f_{\ell
}(y))\right] \,dx dy \cr
&=&- \sqrt{ \frac{\lambda _{\ell }}{2}}\sqrt{\frac{\pi }{2}}\phi (u)\left ( H_{4}(u)+2H_{2}(u)-\frac{3}{2}%
\right ) \frac{1}{4!} \times 8\pi ^{2}\int_{0}^{\pi }\E\left [ \sum_{q=3}^{+\infty}  \Psi^\varepsilon_\ell(\overline x;q) H_{4}(f_{\ell
}(y(\theta)))\right] \sin \theta\, d\theta \\
&=&- \sqrt{ \frac{\lambda _{\ell }}{2}}\sqrt{\frac{\pi }{2}}\phi (u)\left ( H_{4}(u)+2H_{2}(u)-\frac{3}{2}%
\right ) \frac{1}{4!} \times 8\pi ^{2}\int_{0}^{\pi }\E\left [\Psi^\varepsilon_\ell(\overline x;4) H_{4}(f_{\ell
}(y(\theta)))\right] \sin \theta\, d\theta.
\end{eqnarray*}%
The integrand $\E\left [\Psi^\varepsilon_\ell(\overline x;4) H_{4}(f_{\ell
}(y(\theta)))\right] $ can be computed explicitly and it
is easily seen to be absolutely bounded for fixed $\ell $, uniformly over $%
\varepsilon ,$ see Lemma \ref{keylemma}. Hence by Lebesgue Theorem we may exchange the limit and the integral, and we have that
\begin{eqnarray*}
&&\Cov\left (  \sum_{q=3}^{+\infty} \text{proj}[\mathcal L_\ell(u)|q], \mathcal{M}_{\ell }(u)\right )
=\lim_{\varepsilon \rightarrow 0}\Cov\left (  \sum_{q=3}^{+\infty} \text{proj}[\mathcal L_\ell^\varepsilon (u)|q], \mathcal{M}_{\ell }(u)\right )\\
&=&- \sqrt{ \frac{\lambda _{\ell }}{2}}\sqrt{\frac{\pi }{2}}\phi (u)\left ( H_{4}(u)+2H_{2}(u)-\frac{3}{2}%
\right ) \frac{1}{4!} \times 8\pi ^{2}\int_{0}^{\pi }\E\left [\Psi_\ell(\overline x;4) H_{4}(f_{\ell
}(y(\theta)))\right] \sin \theta\, d\theta.
\end{eqnarray*}%
Performing the change of variable $\psi = L\theta$, we can now write
\begin{equation*}
\Cov\left ( \sum_{q=3}^{+\infty} \text{proj}[\mathcal L_\ell(u)|q],\mathcal{M}_{\ell }(u)\right ) =\int_{0}^{L\pi }%
\mathcal{J}_{\ell }(\psi ;4)\sin \frac{\psi }{L}d\psi.
\label{CovDen}
\end{equation*}%
It is now sufficient to notice that
\begin{eqnarray}
\Cov\left (\sum_{q=3}^{+\infty} \text{proj}[\mathcal L_\ell(u)|q],\mathcal{M}_{\ell }\right )
=\int_{0}^{C}\mathcal{J}_{\ell }(\psi ;4)\sin \frac{\psi }{L}d\psi
+2\int_{C}^{L\pi /2}\mathcal{J}_{\ell }(\psi ;4)\sin \frac{\psi }{L}d\psi.
  \label{CovDen2}
\end{eqnarray}
For the first summand in (\ref{CovDen2}) we have thanks to \paref{J1}
\begin{equation*}
\int_{0}^{C}\mathcal{J}_{\ell }(\psi ;4)\sin \frac{\psi }{L}d\psi \leq \ell
\int_{0}^{C}\sin \frac{\psi }{L}d\psi \leq \frac{\ell }{L}\int_{0}^{C}\psi
d\psi =O(1)\text{ , as }\ell \rightarrow \infty;
\end{equation*}%
for the second sum in (\ref{CovDen2}), using Lemma \ref{keylemma}  (precisely \paref{J2}) and
integrating we obtain%
\begin{equation*}
2\int_{C}^{L\pi /2}\mathcal{J}_{\ell }(\psi ;4)\sin \frac{\psi }{L}d\psi = \frac{\pi}{4} \phi(u)^2 \left ( H_{4}(u)+2H_{2}(u)-\frac{3}{2}\right )^2  \log \ell+O(1)
\end{equation*}
giving \paref{cov1}.
\end{proof}

\section*{Appendix}

\subsection{Proof of Lemma \ref{lemmaFubini}}

\begin{proof}
From Proposition \ref{prop_exp} we can write (recall in particular \paref{psi} and \paref{green2}), for every $\varepsilon >0$,
\begin{eqnarray}\label{app1}
\Var\left ( \sum_{q=3}^{+\infty} \int_{\mathbb S^2} \Psi_\ell^\varepsilon(x;q)\,dx  \right ) &=& \int_{\mathbb S^2}\int_{\mathbb S^2}  \mathbb E \left [ \Psi_\ell^\varepsilon(x) \Psi_\ell^\varepsilon(y)\right ]\,dxdy\cr
&& -  \int_{\mathbb S^2} \int_{\mathbb S^2} \E \left [  \Psi_\ell^\varepsilon(x;2) \Psi_\ell^\varepsilon(y;2) \right ]\,dxdy,
\end{eqnarray}
where for the last equality we applied Fubini Theorem, indeed recall \paref{green2} and that for every $x\in \mathbb S^2$
$$
|\Psi^\varepsilon(x)| \le \frac{1}{2\eps} \| f_\ell(x) \|.
$$
Thus we have for \paref{app1}, from the definition of $\Psi^\varepsilon_\ell$ in \paref{psi} and then by isotropy and usual symmetry arguments \cite{Wig09},
\begin{eqnarray}\label{scambio}
\Var\left ( \sum_{q=3}^{+\infty} \int_{\mathbb S^2} \Psi_\ell^\varepsilon(x;q)\,dx  \right ) &=& \int_{\mathbb S^2} \int_{\mathbb S^2} \E\left [ \sum_{q=2}^{+\infty} \Psi_\ell^\varepsilon(x;q) \sum_{q'=2}^{+\infty} \Psi_\ell^\varepsilon(y;q') \right ]\,dx dy\cr
&& -  \int_{\mathbb S^2} \int_{\mathbb S^2} \E \left [  \Psi_\ell^\varepsilon(x;2) \Psi_\ell^\varepsilon(y;2) \right ]\,dxdy\cr
&=& \int_{\mathbb S^2} \int_{\mathbb S^2} \sum_{q=3}^{+\infty} \E\left [ \Psi_\ell^\varepsilon(x;q) \Psi_\ell^\varepsilon(y;q) \right ]\,dx dy\cr
&=& 2\cdot 8\pi^2 \int_0^{\pi/2} \sum_{q=3}^{+\infty} \E\left [  \Psi_\ell^\varepsilon(\overline x;q) \Psi_\ell^\varepsilon(y(\theta);q) \right ]\sin \theta\, d\theta.
\end{eqnarray}
Let us split the integral on the r.h.s. of \paref{scambio} into two terms ($C>0$ is an absolute constant)
\begin{eqnarray}\label{2terms}
&&\int_0^{\pi/2} \sum_{q=3}^{+\infty} \E\left [  \Psi_\ell^\varepsilon(\overline x;q) \Psi_\ell^\varepsilon(y(\theta);q) \right ]\sin \theta\, d\theta\cr
&& =\int_0^{C/\ell} \sum_{q=3}^{+\infty} \E\left [  \Psi_\ell^\varepsilon(\overline x;q) \Psi_\ell^\varepsilon(y(\theta);q) \right ]\sin \theta\, d\theta\cr
&&+\int_{C/\ell}^{\pi/2} \sum_{q=3}^{+\infty} \E\left [  \Psi_\ell^\varepsilon(\overline x;q) \Psi_\ell^\varepsilon(y(\theta);q) \right ]\sin \theta\, d\theta.
\end{eqnarray}
We will separately investigate the two terms on the r.h.s. of \paref{2terms}.
 For the first one, we can write
 \begin{eqnarray}\label{partez1}
 &&\int_0^{C/\ell} \sum_{q=3}^{+\infty} \E\left [  \Psi_\ell^\varepsilon(\overline x;q) \Psi_\ell^\varepsilon(y(\theta);q) \right ]\sin \theta\, d\theta \cr
 &&=   \int_0^{C/\ell} K^\varepsilon_\ell(\overline x, y(\theta))\sin \theta\,d\theta \cr
 && - \int_0^{C/\ell} \E\left [  \Psi^\varepsilon_\ell(\overline x;2) \Psi^\varepsilon_\ell(y(\theta);2) \right ] \sin \theta\,d\theta,
   \end{eqnarray}
   where for $x,y\in \mathbb S^2$, $K^\varepsilon_\ell(x,y) := \mathbb E[\Psi^\varepsilon_\ell(x)\Psi^\varepsilon_\ell(y)]$ is the $\varepsilon$-approximation of the so-called two point correlation function (see \cite{Wig09})
  $$
  K_\ell(x,y):=p_{(f_\ell(x),f_\ell(y))}(u,u)\mathbb E[\| \nabla f_\ell(x) \| \|\nabla f_\ell(y)\| | f_\ell(x)=f_\ell(y)=u],\quad x,y\in \mathbb S^2,
  $$
  $p_{(f_\ell(x),f_\ell(y))}$ denoting the density of the Gaussian vector ${(f_\ell(x),f_\ell(y))}$.
We can use Lemma 3.4 so that
\begin{equation*}
\lim_{\varepsilon\to 0} \int_0^{C/\ell} K^\varepsilon_\ell(\overline x, y(\theta))\sin \theta\,d\theta = \int_0^{C/\ell} K_\ell(\overline x, y(\theta))\sin \theta\,d\theta
\end{equation*}
and then Corollary 3.5 in \cite{Wig09} entailing that
\begin{equation}\label{stimaI}
\int_0^{C/\ell} K_\ell(\overline x, y(\theta))\sin \theta\,d\theta =O\left (1 \right ).
\end{equation}
From  \paref{green2} it is immediate to show that the integrand of the second term on the r.h.s. of \paref{partez1} is $O(\ell^2)$ uniformly in $\varepsilon$ and so
$$
\lim_{\varepsilon\to 0} \int_0^{C/\ell} \E\left [  \Psi^\varepsilon_\ell(\overline x;2) \Psi^\varepsilon_\ell(y(\theta);2) \right ] \sin \theta\,d\theta =  \int_0^{C/\ell} \E\left [  \Psi_\ell(\overline x;2) \Psi_\ell(y(\theta);2) \right ] \sin \theta\,d\theta = O(1)
$$
that together with \paref{stimaI} gives that the l.h.s. of \paref{partez1} is $O(1)$.
 The second term on the r.h.s. of (\ref{2terms}) is more delicate to deal with (we will show that we can exchange integral and series), as follows.
\begin{eqnarray}\label{conto1}
&&\int_{C/\ell}^{\pi/2}  \left |\E\left [  \Psi^\varepsilon_\ell(\overline x;q) \Psi^\varepsilon_\ell(y(\theta);q) \right ]\right | \sin\theta\,d\theta \cr
&&\le \frac{\ell (\ell
+1)}{2}\sum_{2a + 2b + c = q}\sum_{2a' + 2b' + c'=q} \left | \frac{\alpha
_{2a,2b}\beta^\eps _{c}(u)}{(2a)!(2b)!c!}\frac{\alpha
_{2a',2b'}\beta^\eps _{c'}(u)}{(2a')!(2b')!(c')!}\right | \times \cr
&&\times \int_{C/\ell}^{\pi/2} \Big | \mathbb E\Big[H_{c}(f_{\ell }(\overline x))H_{2a}(\widetilde \partial
_{1;x}f_{\ell }(\overline x))H_{2b}(\widetilde \partial
_{2;x}f_{\ell }(\overline x)) \times \cr
&&\times H_{c'}(f_{\ell }(y(\theta)))H_{2a}(\widetilde \partial
_{1;x}f_{\ell }(y(\theta))) H_{2b}(\widetilde \partial
_{2;x}f_{\ell }(y(\theta))) \Big ] \Big |\sin \theta\,d\theta  \cr
&&\le   \ell (\ell
+1) \sum_{2a + 2b + c = q}\sum_{2a' + 2b' + c'=q} \left | \frac{\alpha
_{2a,2b}\beta^\eps _{c}(u)}{(2a)!(2b)!c!} \right | \left |\frac{\alpha
_{2a',2b'}\beta^\eps _{c'}(u)}{(2a')!(2b')!(c')!} \right |  \times\cr
&&\times \Big | \mathcal V(2a,2b,c,2a',2b',c')   \Big |,
\end{eqnarray}
where $\mathcal V(2a,2b,c,2a',2b',c') $ is the sum of no more than $q!$ terms of the form
\begin{eqnarray}\label{explicit1}
&&\left ( \frac{2}{\ell(\ell+1)}\right )^{\frac12 m_2 + m_3 + m_4}\int_{C/\ell}^{\pi/2}  P_\ell(\cos \theta)^{m_1} (P'_\ell(\cos \theta)\sin \theta)^{m_2} \times \cr
&& \times (P_\ell(\cos \theta)\cos \theta - P_\ell^{''}(\cos \theta)\sin^2\theta)^{m_3}  P'_\ell (\cos \psi/L)^{m_4}\sin\theta\,d\theta,
\end{eqnarray}
where $m_1, \dots, m_4 \ge 0$ and $m_1+m_2+m_3+m_4 = q$. Since $C/\ell <\theta < \pi/2$ we have that the absolute value of each (properly normalized) factor in \paref{explicit1} is less than $1-\delta$ for some small $\delta >0$, see i.e., the expressions for $P_{\ell}$, $P'_{\ell},$ $P''_{\ell}$ which are proved in \cite{CMW}, Lemma B3 and reported in \cite{MRW}, Appendix A.   Hence can write from \paref{conto1}, taking into account \paref{explicit1},
\begin{eqnarray*}
&&\sum_{q=3}^{+\infty}\int_{C/\ell}^{\pi/2}  \left | \E\left [  \Psi^\varepsilon_\ell(\overline x;q) \Psi^\varepsilon_\ell(y(\theta);q) \right ] \right |\sin\theta\,d\theta  \cr
&&\le   c\cdot  \ell (\ell
+1) \sum_{q=3}^{+\infty}\sum_{2a + 2b + c = q}\sum_{2a' + 2b' + c'=q} \left | \frac{\alpha
_{2a,2b}\beta^\eps _{c}(u)}{(2a)!(2b)!c!} \right | \left |\frac{\alpha
_{2a',2b'}\beta^\eps _{c'}(u)}{(2a')!(2b')!(c')!} \right |  q! (1-\delta)^{q},
\end{eqnarray*}
for some $c>0$.
Repeating the same argument as in the proof of Lemma 3.5 in \cite{DNPR2016}, we get
\begin{eqnarray*}\label{semifinal}
&&\sum_{q=3}^{+\infty}  q! (1-\delta)^q \sum_{2a + 2b + c = q}\sum_{2a' + 2b' + c'=q} \left | \frac{\alpha
_{2a,2b}\beta^\eps _{c}(u)}{(2a)!(2b)!c!} \right | \left |\frac{\alpha
_{2a',2b'}\beta^\eps _{c'}(u)}{(2a')!(2b')!(c')!} \right |  \cr
&& \le \sum_{a,b,c,a',b',c'\ge 0} \frac{\alpha
_{2a,2b}^2\beta^\eps _{c}(u)^2}{(2a)!(2b)!c!} \frac{(2a+2b+c)!}{(2a)!(2b)!c!} \sqrt{1-\delta}^{2a+2b+c+2a'+2b'+c'} < +\infty.
\end{eqnarray*}
Indeed, note that the map
$$
(a,b,c)\mapsto \frac{\alpha
_{2a,2b}^2\beta^\eps _{c}(u)^2}{(2a)!(2b)!c!}
$$
is  bounded uniformly over $\eps$: see \paref{alpha} and recall that there exists $C>0$ s.t. for every $k\in \mathbb N$ and $u\in \mathbb R$
\begin{equation*}
|H_k(u)| \gamma(u) \le C \sqrt{k!}
\end{equation*}
 immediately implying (see the definition of $\beta^\eps_{\cdot}$ in \paref{beta}) that for every $c\in \mathbb N$ and $\eps>0$
$$
\frac{\beta_c^\eps (u)^2}{ c!} \le C.
$$
We have just proved that
\begin{eqnarray*}
&&\int_{C/\ell}^{\pi/2} \sum_{q=3}^{+\infty} \E\left [  \Psi_\ell^\varepsilon(\overline x;q) \Psi_\ell^\varepsilon(y(\theta);q) \right ]\sin \theta\, d\theta \cr
&&=\sum_{q=3}^{+\infty} \int_{C/\ell}^{\pi/2} \E\left [  \Psi_\ell^\varepsilon(\overline x;q) \Psi_\ell^\varepsilon(y(\theta);q) \right ]\sin \theta\, d\theta
\end{eqnarray*}
and moreover
\begin{eqnarray*}
&&\lim_{\varepsilon\to 0} \sum_{q=3}^{+\infty} \int_{C/\ell}^{\pi/2} \E\left [  \Psi_\ell^\varepsilon(\overline x;q) \Psi_\ell^\varepsilon(y(\theta);q) \right ]\sin \theta\, d\theta\cr
&& = \sum_{q=3}^{+\infty} \int_{C/\ell}^{\pi/2} \E\left [  \Psi_\ell(\overline x;q) \Psi_\ell(y(\theta);q) \right ]\sin \theta\, d\theta.
\end{eqnarray*}
which is what we were looking for.
 \end{proof}

\subsection{Proof of Lemma \ref{keylemma}}

The projection of the boundary length on the fourth-order
chaos is
\begin{eqnarray}\label{fourthchaosdom}
\text{proj}[{\mathcal{L}}_{\ell }(u)|4] =\int_{\mathbb{S}^{2}} & \Psi _{\ell
}(x;4)\,dx=\sqrt{\frac{\ell (\ell +1)}{2}}\frac{\alpha _{00}\beta _{4}(u)}{4!}%
\int_{\mathbb{S}^{2}}H_{4}(f_{\ell }(x))\,dx\cr
&+\frac{\alpha _{20}\beta _{2}(u)}{%
2!2!}\int_{\mathbb{S}^{2}}H_{2}(f_{\ell }(x))H_{2}\left(\frac{\partial
_{1;x}f_{\ell }(x)}{\sqrt{\ell (\ell +1)/2}}\right)dx. \cr
&+\frac{\alpha _{40}\beta _{0}(u)}{4!}\int_{\mathbb{S}^{2}}H_{4}\left(\frac{\partial
_{1;x}f_{\ell }(x)}{\sqrt{\ell (\ell +1)/2}}\right)dx\cr
&+\frac{\alpha
_{22}\beta _{0}(u)}{2!2!}\int_{\mathbb{S}^{2}}H_{2}\left(\frac{\partial
_{1;x}f_{\ell }(x)}{\sqrt{\ell (\ell +1)/2}}\right)H_{2}\left(\frac{\partial
_{2;x}f_{\ell }(x)}{\sqrt{\ell (\ell +1)/2}}\right)dx\cr
&+\frac{\alpha _{02}\beta _{2}(u)}{2!2!}\int_{\mathbb{S}%
^{2}}H_{2}(f_{\ell }(x))H_{2}\left(\frac{\partial
_{2;x}f_{\ell }(x)}{\sqrt{\ell (\ell +1)/2}}\right)dx\cr
&+ \frac{\alpha _{04}\beta _{0}(u)}{4!}\int_{\mathbb{S}^{2}}H_{4}\left(\frac{\partial
_{2;x}f_{\ell }(x)}{\sqrt{\ell (\ell +1)/2}}\right)dx\cr
& =: \int_{\mathbb S^2} (A_\ell(x) + B_\ell(x) + C_\ell(x) + D_\ell(x) + E_\ell(x) + F_\ell(x))dx.
\end{eqnarray}
Let us set
$$
\mathcal M_\ell(u) = \sqrt{ \frac{\lambda _{\ell }}{2}}\sqrt{\frac{\pi }{2}}\phi (u)\left ( H_{4}(u)+2H_{2}(u)-\frac{3}{2}%
\right ) \frac{1}{4!}\int_{\mathbb{S}^{2}}H_{4}(f_{\ell }(x))dx = : \int_{\mathbb S^2} M_\ell(x;u)\,dx.
$$
\begin{proof}[Proof of Lemma \ref{keylemma}]
Repeating the same argument as in the proof of Proposition 3.1 in \cite{MRW} we prove \paref{J1}.
%
%
For the computations to follow, recall that  $\overline x = (0,0)$ and $y(\theta) = (\theta, 0)$. It is sufficient to focus on the terms $A_\ell$, $B_\ell$ and $C_\ell$, as in \cite[Proposition 3.1]{MRW}. 
An application of Diagram Formula gives
\begin{eqnarray}
\E\left [ A_\ell(\overline x) M_\ell(y(\theta);u)    \right] &=& \frac{\ell(\ell+1)}{2} \sqrt{\frac{\pi}{2}} \frac{\alpha_{00}\beta_4(u)}{4!} \phi(u) \left (H_4(u) + 2H_2(u) -\frac{3}{2}\right ) P_\ell(\cos \theta)^4, \cr
\E\left [ B_\ell(\overline x) M_\ell(y(\theta);u)    \right] &=&  \frac{\ell(\ell+1)}{2}\sqrt{\frac{\pi}{2}} \frac{\alpha_{20}\beta_2(u)}{2! 2!} \phi(u) \left (H_4(u) + 2H_2(u) -\frac{3}{2}\right )\times \cr
&\times& \frac{2}{\ell(\ell+1)}P_\ell(\cos \theta)^2 (P_\ell'(\cos \theta) \sin\theta)^2, \cr
\E\left [ C_\ell(\overline x) M_\ell(y(\theta);u)    \right] &=&  \frac{\ell(\ell+1)}{2}\sqrt{\frac{\pi}{2}} \frac{\alpha_{40}\beta_0(u)}{4!} \phi(u) \left (H_4(u) + 2H_2(u) -\frac{3}{2}\right )\times \cr
& \times& \left( \frac{2}{\ell(\ell+1)}\right )^2(P_\ell'(\cos \theta)\sin\theta)^4.
\end{eqnarray}
Analogously to the proof of Proposition 3.1 in \cite{MRW} we have
\begin{eqnarray}\label{1}
8\pi^2\E\left [ A_\ell(\overline x) M_\ell(y(\theta);u)    \right] &=& 8\pi^2 \frac{\ell(\ell+1)}{2} \sqrt{\frac{\pi}{2}} \frac{\alpha_{00}\beta_4(u)}{4!} \phi(u) \left (H_4(u) + 2H_2(u) -\frac{3}{2}\right ) P_\ell(\cos \theta)^4 \cr
 &= & 8\pi^2 \frac{\ell(\ell+1)}{2} \sqrt{\frac{\pi}{2}} \frac{\alpha_{00}\beta_4(u)}{4!} \phi(u) \left (H_4(u) + 2H_2(u) -\frac{3}{2}\right ) \times \cr
 &&\left ( \sqrt{\frac{2}{\pi\ell \sin \psi/L}}\left( \sin(\psi + \pi/4) + O\left(\frac{1}{\psi}\right)   \right ) \right )^4 \cr
 &=&  8\pi^2 \frac{\ell(\ell+1)}{2} \sqrt{\frac{\pi}{2}} \frac{\alpha_{00}\beta_4(u)}{4!} \phi(u) \left(H_4(u) + 2H_2(u) -\frac{3}{2}\right) \times \cr
 &\times& \frac{2^2}{\pi^2\ell^2 \sin^2 \psi/L} \sin(\psi + \pi/4)^4 + O\left(\frac{1}{(\psi \sin^2 \psi/L)}\right)      \cr
 & =&   8\pi^2 \frac{\ell(\ell+1)}{2} \sqrt{\frac{\pi}{2}} \frac{\alpha_{00}\beta_4(u)}{4!} \phi(u) \left (H_4(u) + 2H_2(u) -\frac{3}{2}\right )  \frac{2^2}{\pi^2\ell^2 \sin^2 \psi/L} \times \cr
 &\times& \left (\frac38  - \frac18\cos 4\psi +\frac12 \sin 2\psi    \right ) + O\left(\frac{1}{(\psi \sin^2 \psi/L)}\right)      .
\end{eqnarray}
Likewise
\begin{eqnarray}\label{2}
&&8\pi ^{2}\mathbb{E}\left[ B_{\ell }(\overline{x})M_{\ell }(y(\theta ))\right]\cr
&=&8\pi^2 \frac{\ell(\ell+1)}{2}\sqrt{\frac{\pi}{2}} \frac{\alpha_{20}\beta_2(u)}{2! 2!} \phi(u) \left (H_4(u) + 2H_2(u) -\frac{3}{2}\right )\frac{2}{\ell(\ell+1)}P_\ell(\cos \theta)^2 (P_\ell'(\cos \theta) \sin\theta)^2\cr
&=&8\pi^2 \frac{\ell(\ell+1)}{2}\sqrt{\frac{\pi}{2}} \frac{\alpha_{20}\beta_2(u)}{2! 2!} \phi(u) \left (H_4(u) + 2H_2(u) -\frac{3}{2}\right )\frac{2}{\ell(\ell+1)}\times \cr
&\times &\left[ \sqrt{\frac{2}{\pi \ell \sin \frac{\psi }{L}}}%
\left(\sin (\psi +\frac{\pi }{4})+O\left(\frac{1}{\psi }\right)\right)\right] ^{2}\left[ \sqrt{%
\frac{2}{\pi \ell \sin ^{3}\frac{\psi }{L}}}\left(\ell \sin \left(\psi -\frac{\pi }{4}%
\right)+O(1)\right)\sin \frac{\psi }{L}\right] ^{2}\cr
&=&8\pi^2 \frac{\ell(\ell+1)}{2}\sqrt{\frac{\pi}{2}} \frac{\alpha_{20}\beta_2(u)}{2! 2!} \phi(u) \left (H_4(u) + 2H_2(u) -\frac{3}{2}\right )\frac{2}{\ell(\ell+1)}\times \cr
&\times& \frac{2}{\pi \ell \sin \frac{\psi }{L}}\sin ^{2}(\psi +%
\frac{\pi }{4})\frac{2}{\pi \ell \sin \frac{\psi }{L}}\ell ^{2}\sin
^{2}\left(\psi -\frac{\pi }{4}\right)+O\left(\frac{1}{\ell \sin ^{2}\frac{\psi }{L}}\right)\cr
&=&8\pi^2 \frac{\ell(\ell+1)}{2}\sqrt{\frac{\pi}{2}} \frac{\alpha_{20}\beta_2(u)}{2! 2!} \phi(u) \left (H_4(u) + 2H_2(u) -\frac{3}{2}\right )\frac{2}{\ell(\ell+1)}\times \cr
&\times &\frac{4}{\pi^2 \sin^2 \frac{\psi }{L}} \frac18 (1+\cos 4\psi) +O\left(\frac{1}{\ell \sin ^{2}\frac{\psi }{L}}\right).
\end{eqnarray}%
Finally
\begin{eqnarray}\label{3}
8\pi ^{2}\mathbb{E}\left[ C_{\ell }(\overline{x})M_{\ell }(y(\theta ))\right]
&=&8\pi^2  \frac{\ell(\ell+1)}{2}\sqrt{\frac{\pi}{2}} \frac{\alpha_{40}\beta_0(u)}{4!} \phi(u) \left (H_4(u) + 2H_2(u) -\frac{3}{2}\right )\times \cr
&\times& \left( \frac{2}{\ell(\ell+1)}\right )^2( P_{\ell }^{\prime }(\cos \theta )\sin \theta
 )^{4}\cr
&=& 8\pi^2  \frac{\ell(\ell+1)}{2}\sqrt{\frac{\pi}{2}} \frac{\alpha_{40}\beta_0(u)}{4!} \phi(u) \left (H_4(u) + 2H_2(u) -\frac{3}{2}\right )\left( \frac{2}{\ell(\ell+1)}\right )^2\times \cr
&\times&  \left ( \sqrt{\frac{2%
}{\pi \ell \sin ^{3}\frac{\psi }{L}}}\left(\ell \sin \left(\psi -\frac{\pi }{4}%
\right)+O(1)\right)\sin \frac{\psi }{L}\right )  ^{4}\cr
&=& 8\pi^2  \frac{\ell(\ell+1)}{2}\sqrt{\frac{\pi}{2}} \frac{\alpha_{40}\beta_0(u)}{4!} \phi(u) \left (H_4(u) + 2H_2(u) -\frac{3}{2}\right )\left( \frac{2}{\ell(\ell+1)}\right )^2\times \cr
&\times& \frac{2^{2}}{\pi
^{2}\ell ^{2}\sin ^{2}\frac{\psi }{L}}\ell ^{4}\sin ^{4}\left(\psi -\frac{\pi }{4}%
\right)+O\left(\frac{1}{\ell \sin ^{2}\frac{\psi }{L}}\right)\cr
&=& 8\pi^2  \frac{\ell(\ell+1)}{2}\sqrt{\frac{\pi}{2}} \frac{\alpha_{40}\beta_0(u)}{4!} \phi(u) \left (H_4(u) + 2H_2(u) -\frac{3}{2}\right )\left( \frac{2}{\ell(\ell+1)}\right )^2\times \cr
&\times&  \frac{2^{2}}{\pi
^{2}\ell ^{2}\sin ^{2}\frac{\psi }{L}}\ell ^{4}\left[ \frac{3}{8}-\frac{1}{8%
}\cos 4\psi -\frac{1}{2}\sin 2\psi \right] +O\left(\frac{1}{\ell \sin ^{2}\frac{%
\psi }{L}}\right).
\end{eqnarray}
Plugging \paref{1}, \paref{2} and \paref{3} all together we find
\begin{eqnarray}
\mathcal J_\ell(\psi; 4) &= & \frac{1}{L} \sqrt{\frac{\pi}{2}}\phi(u) (H_4(u) + 2H_2(u) -\frac32) \times \cr
&&\Big \lbrace 8\pi^2 \frac{\ell(\ell+1)}{2} \frac{\alpha_{00}\beta_4(u)}{4!}    \frac{2^2}{\pi^2\ell^2 \sin^2 \psi/L} \times \cr
 &\times& \left (\frac38  - \frac18\cos 4\psi +\frac12 \sin 2\psi    \right ) + O\left(\frac{1}{(\psi \sin^2 \psi/L)}\right)      \cr
 && + 8\pi^2 \frac{\ell(\ell+1)}{2} \frac{\alpha_{20}\beta_2(u)}{2! 2!} \frac{2}{\ell(\ell+1)}\times \cr
&\times& \frac{4}{\pi^2 \sin^2 \frac{\psi }{L}} \frac18 (1+\cos 4\psi) +O\left(\frac{1}{\ell \sin ^{2}\frac{\psi }{L}}\right) \cr
&& + 8\pi^2  \frac{\ell(\ell+1)}{2} \frac{\alpha_{40}\beta_0(u)}{4!} \left( \frac{2}{\ell(\ell+1)}\right )^2\times \cr
&\times&  \frac{2^{2}}{\pi
^{2}\ell ^{2}\sin ^{2}\frac{\psi }{L}}\ell ^{4}\left[ \frac{3}{8}-\frac{1}{8%
}\cos 4\psi -\frac{1}{2}\sin 2\psi \right]  \Big \rbrace \cr
&= &\frac{1}{L} \sqrt{\frac{\pi}{2}}\phi(u) (H_4(u) + 2H_2(u) -\frac32) \times \cr
&\times& \Big \lbrace  \frac{2}{3} \alpha_{00}\beta_4(u)    \frac{1}{ \sin^2 \psi/L} \left (\frac38  - \frac18\cos 4\psi +\frac12 \sin 2\psi    \right ) + O\left(\frac{1}{(\psi \sin^2 \psi/L)}\right)      \cr
 && +  \alpha_{20}\beta_2(u)  \frac{1}{ \sin^2 \frac{\psi }{L}}  (1+\cos 4\psi) +O\left(\frac{1}{\ell \sin ^{2}\frac{\psi }{L}}\right) \cr
&& + \frac83 \alpha_{40}\beta_0(u) \frac{1}{\sin ^{2}\frac{\psi }{L}} \left[ \frac{3}{8}-\frac{1}{8%
}\cos 4\psi -\frac{1}{2}\sin 2\psi \right]  \Big \rbrace \cr
&=&\sqrt{\frac{\pi}{2}}\phi(u) (H_4(u) + 2H_2(u) -\frac32) \times \cr
&\times& \Big \lbrace \frac{1}{L \sin^2 \psi/L} (\frac14 \alpha_{00}\beta_4(u) + \alpha_{20}\beta_2(u) + \alpha_{40}\beta_0(u))\cr
&&+ \frac{\cos 4\psi}{L \sin^2\psi/L}(-\frac{1}{12}\alpha_{00}\beta_4(u)+\alpha_{20}\beta_2(u) -\frac13 \alpha_{40}\beta_0(u))\cr
 && + \frac{\sin 2\psi}{L\sin^2\psi/L}(\frac13 \alpha_{00}\beta_4(u) -\frac43 \alpha_{40}\beta_0(u))
  \Big \rbrace + O\left(\frac{1}{\psi \sin^2 \psi/L}\right) +O\left(\frac{1}{\ell \sin ^{2}\frac{\psi }{L}}\right)  \cr
&=& \sqrt{\frac{\pi}{2}}\phi(u) (H_4(u) + 2H_2(u) -\frac32) \times \cr
&\times& \Big \lbrace \frac{1}{\psi \sin \psi/L} (\frac14 \alpha_{00}\beta_4(u) + \alpha_{20}\beta_2(u) + \alpha_{40}\beta_0(u))\cr
&&+ \frac{\cos 4\psi}{\psi \sin\psi/L}(-\frac{1}{12}\alpha_{00}\beta_4(u)+\alpha_{20}\beta_2(u) -\frac13 \alpha_{40}\beta_0(u))\cr
 && + \frac{\sin 2\psi}{\psi \sin\psi/L}(\frac13 \alpha_{00}\beta_4(u) -\frac43 \alpha_{40}\beta_0(u))
  \Big \rbrace  + O\left(\frac{1}{\psi \sin^2 \psi/L}\right) +O\left(\frac{1}{\ell \sin ^{2}\frac{\psi }{L}}\right) \cr
  &= & \sqrt{\frac{\pi}{2}}\phi(u) (H_4(u) + 2H_2(u) -\frac32) \times \cr&\times &\Big \lbrace \frac{1}{\psi \sin \psi/L} \sqrt{\frac{\pi}{2}}\phi(u) \frac14 (H_4(u) + 2H_2(u) -\frac32)  \Big \rbrace + a(u) \frac{\cos 4\psi}{\psi \sin \psi/L} + b(u) \frac{\sin 2\psi}{\psi \sin \psi/L} \cr
  &&+ O\left(\frac{1}{\psi \sin^2 \psi/L}\right) +O\left(\frac{1}{\ell \sin ^{2}\frac{\psi }{L}}\right) \cr
    &=&  \frac14 \frac{\pi}{2}\phi(u)^2 (H_4(u) + 2H_2(u) -\frac32)^2 \times \cr&\times &\Big \lbrace \frac{1}{\psi \sin \psi/L}   \Big \rbrace + a(u) \frac{\cos 4\psi}{\psi \sin \psi/L} + b(u) \frac{\sin 2\psi}{\psi \sin \psi/L} \cr
  &&+ O\left(\frac{1}{\psi \sin^2 \psi/L}\right) +O\left(\frac{1}{\ell \sin ^{2}\frac{\psi }{L}}\right).
  \end{eqnarray}
as claimed.
\end{proof}

\

{\sc Dipartimento di Matematica, Universit\'a di Roma ``Tor Vergata"}

\emph{marinucc[at]mat.uniroma2.it} (Corresponding author)

\medskip

{\sc Dipartimento di Matematica, Universit\'a di Pisa }

\emph{maurizia.rossi[at]unipi.it}

\end{document}